\documentclass[
 superscriptaddress,
 numerical,
 amsmath,amssymb,
 aps,
 prx,
 longbibliography,
 floatfix,
 notitlepage,
 nofootinbib,
 twocolumn
]{revtex4-2}

\usepackage{times}
\usepackage{graphicx}
\usepackage{bm}
\usepackage{mathtools}
\usepackage{braket}
\usepackage{ulem}
\usepackage[dvipsnames]{xcolor}
\usepackage[
  colorlinks=true,
  urlcolor=Blue,
  linkcolor=BrickRed,
  citecolor=Blue,
  linktoc=all,
  bookmarksnumbered=true,
  bookmarksopen=true
]{hyperref}
\usepackage{fix-cm}
\usepackage{tikz}
\usetikzlibrary{arrows.meta,calc}
\usepackage{comment}

\makeatletter
\let\original@opargbegintheorem\@opargbegintheorem
\def\@opargbegintheorem#1#2#3{%
  \original@opargbegintheorem{#1}{#2\normalfont}{#3}%
}
\makeatother

\newtheorem{theorem}{Theorem}
\newtheorem{proposition}[theorem]{Proposition}
\newtheorem{lemma}[theorem]{Lemma}

\newenvironment{proof}{\par\noindent\textit{Proof.}\ }{\hfill$\square$\par}

\DeclareMathOperator{\supp}{supp}
\DeclareMathOperator{\Tr}{Tr}
\newcommand{\id}{\operatorname{id}}
\newcommand{\Ad}{\operatorname{Ad}}

\newcommand{\cC}{\mathcal C}
\newcommand{\cE}{\mathcal E}
\newcommand{\cM}{\mathcal M}
\newcommand{\cN}{\mathcal N}
\newcommand{\cR}{\mathcal R}
\newcommand{\cD}{\mathcal D}
\newcommand{\cQ}{\mathcal Q}
\newcommand{\cG}{\mathcal G}
\newcommand{\cT}{\mathcal T}
\newcommand{\eps}{\varepsilon}
\newcommand{\norm}[1]{\left\lVert #1\right\rVert}
\newcommand{\qparams}[3]{\left[\!\left[#1,#2,#3\right]\!\right]}
\renewcommand{\emph}{\textit}

\definecolor{ForestGreen}{rgb}{0.13, 0.55, 0.13}

\begin{document}

\title{Coherent error threshold for quantum LDPC codes}

\author{Zhengyi Han}
\thanks{These authors contributed equally to this work.}
\affiliation{Yau Mathematical Sciences Center, Tsinghua University, Beijing 100084, China}

\author{Yuanchen Zhao}
\thanks{These authors contributed equally to this work.}
\affiliation{State Key Laboratory of Low Dimensional Quantum Physics, Department of Physics, Tsinghua University, Beijing, 100084, China}

\author{Yijia Xu}
\affiliation{Joint Center for Quantum Information and Computer Science, University of Maryland, College Park,
Maryland 20742, USA}

\author{Yixu Wang}
\affiliation{Shanghai Institute for Mathematics and Interdisciplinary Sciences (SIMIS), Shanghai 200433, China}

\author{Zi-Wen Liu}
\email{zwliu0@tsinghua.edu.cn}
\affiliation{Yau Mathematical Sciences Center, Tsinghua University, Beijing 100084, China}

\date{\today}

\begin{abstract}

A key appeal of quantum low-density parity check (qLDPC) codes is their ability to suppress stochastic Pauli noise below nonzero thresholds. Coherent errors are fundamentally different: they produce superpositions of error patterns whose amplitudes can interfere even after syndrome measurement. Rigorous understanding of coherent errors remains limited. Here we show that general qLDPC codes admit a nonzero code capacity threshold against local coherent noise and more generally local channel noise. For any family of qLDPC codes with distance $d=\Omega(\log n)$, we show that there is a constant noise strength below which the logical recovery error in diamond distance decays exponentially with the code distance. The result is established for optimal recovery as well as the minimum-weight decoder. The key technical ingredient is what we call a \emph{cluster resummation}: rather than bounding superposed error configurations one by one, we isolate a large connected error cluster in the channel expansion and exactly resum all errors disconnected from it before taking norms. Standard cluster counting then yields exponential suppression.
This work resolves a longstanding challenge in fault tolerance theory, providing general robustness guarantees for qLDPC codes against coherent noise and laying a rigorous foundation for future studies of fault-tolerant quantum technologies.
\end{abstract}

\maketitle


\section{Introduction}

Physical quantum computers are inevitably affected by noise and errors arising from decoherence and imperfect control. Without adequate protection, the resulting errors can accumulate during a computation and compromise the reliable execution
of large-scale quantum algorithms
\cite{Shor1995,Preskill1998}. Realizing scalable quantum computers requires fault-tolerant quantum computation (FTQC) with sufficiently low overhead~\cite{Shor1996,Preskill1998}. Central to this goal is quantum error correction (QEC), which encodes logical qubits redundantly across physical qubits and uses error syndromes to detect and correct errors without revealing the encoded information~\cite{Shor1995,Steane1996,CalderbankShor1996,KnillLaflamme1997}.
A crucial question for scalability is whether increasing the code size can make the logical error rate arbitrarily small at a fixed nonzero physical noise strength~\cite{DennisEtAl2002,Gottesman2014}. A fault-tolerance threshold extends this guarantee to entire quantum computations, with rigorous results established for concatenated codes~\cite{TerhalBurkard2005,AliferisGottesmanPreskill2006,AharonovBenOr2008,YamasakiKoashi2024}.

Quantum low-density parity-check (qLDPC) codes provide a leading framework in this context. These are stabilizer codes admitting generating sets in which each check acts on only a constant number of physical qubits and each qubit participates in only a constant number of checks, independently of the code size~\cite{Gottesman1997,BreuckmannEberhardt2021}. This class encompasses both topological codes and nongeometric constructions, and recent breakthroughs have produced asymptotically good qLDPC families~\cite{PanteleevKalachev2022,LeverrierZemor2022,DinurEtAl2023}. Understanding whether and under what conditions these codes admit nonzero thresholds under realistic noise is a problem central to their use in scalable fault-tolerant quantum computation.

The threshold theory for stochastic noise is well established. In the qLDPC setting, both memory-level thresholds~\cite{DennisEtAl2002,KovalevPryadko2013,DumerKovalevPryadko2015,kovalevNumericalAnalyticalBounds2018,chubbStatisticalMechanicalModels2021} and thresholds for full fault-tolerant quantum computation~\cite{Gottesman2014,NguyenPattison2025,TamiyaEtAl2026} have been extensively studied using analytical and numerical methods. In particular, stochastic Pauli noise models fit within the stabilizer formalism, permitting large-scale numerical simulation of stabilizer circuits using the Gottesman--Knill theorem~\cite{Gottesman1997,AaronsonGottesman2004}.

The case of coherent noise is fundamentally different and much less understood. Miscalibrations of control phases produce unitary errors~\cite{FengEtAl2016,PollorenoYoung2022,BravyiEtAl2018}. Such noise contains
quantum superpositions of {error} patterns. Although syndrome measurement removes coherence between distinct syndrome sectors, amplitudes of errors with the same syndrome can still interfere. Consequently, stochastic threshold arguments based on probabilities of error patterns do not apply. 
Generic non-Clifford errors also preclude the use of efficient stabilizer simulation.
Several lines of work have established a partial understanding of robustness under coherent noise in specific settings.
Memory-level analyses of stabilizer codes have characterized the suppression of logical coherence by error correction under suitable assumptions~\cite{BealeEtAl2018,IversonPreskill2020}.
Tensor network simulations have also demonstrated coherent error suppression in surface codes~\cite{DarmawanPoulin2017}. Further numerical and statistical-mechanical analyses based on tensor network methods or free-fermion mappings available for specific noise models  have provided evidence for  coherent noise thresholds in surface and toric codes~\cite{BravyiEtAl2018,VennBeri2020,VennEtAl2023,BehrendsBeri2025a,BehrendsBeri2025b,MartonAsboth2023,YangEtAl2026}.
More recently, fault-tolerant quantum computation with constant qubit overhead has been established under general circuit-level noise, including coherent errors, using constant-rate qLDPC codes with linear distance~\cite{ChristandlFawziGoswami2025}. However, the argument relies on linear distance and does not cover many important qLDPC constructions such as topological codes~\cite{Kitaev2003,DennisEtAl2002,BombinMartinDelgado2006Color,BombinMartinDelgado2006Rate,BreuckmannTerhal2016}, quantum expander codes~\cite{LeverrierTillichZemor2015}, and known qLTC~\cite{DinurLinVidick2024}. Such codes also underpin important fault-tolerant schemes, including some of the lowest overhead schemes
currently known~\cite{NguyenPattison2025,TamiyaEtAl2026}.
A general rigorous understanding of coherent noise thresholds for qLDPC codes remains elusive.

Here we prove general nonzero code capacity thresholds for any family of
stabilizer qLDPC codes with distance $d = \Omega(\log n)$ under unknown local coherent noise. The noise is represented by an arbitrary constant-depth circuit of gates $e^{-\mathrm{i}\eta h}$, where the Hermitian generators $h$ have uniformly bounded operator norm and support size. After the noise, all stabilizer checks are measured perfectly, followed by a completely positive trace-preserving (CPTP) recovery channel.
We first establish the threshold for optimal recovery, which defines the best achievable error correction performance and provides a fundamental information-theoretic benchmark. We further show that the threshold persists under minimum-weight decoding, demonstrating that it is not merely an idealized phenomenon but can also hold for natural concrete decoding schemes without prior knowledge of the coherent noise.

Our analysis adopts an approximate quantum error correction (AQEC) viewpoint, bounding the optimal recovery error through the complementary channel, which captures the logical information leaked to the environment~\cite{BenyOreshkov2010,FaistEtAl2020,KongLiu2022}. AQEC measures have been used to characterize decoder-independent error thresholds~\cite{ZhaoLiu2024,ColmenarezEtAl2024,Lee2025,NiwaLee2025,ColmenarezKimMuller2025,Kim2026}; here we use the diamond distance, which is well suited to coherent logical noise~\cite{KretschmannSchlingemannWerner2008}. Specifically, we prove that there exists a constant $\eta_{\rm th}>0$, independent of $n$, such that for $|\eta|<\eta_{\rm th}$, the logical error measured in diamond distance is exponentially small in the code distance. Consequently, logarithmic code distance suffices to suppress the logical error.
The proof roughly goes as follows. We expand the coherent error channel into channel error configurations and show that subdistance connected components carry no logical information after syndrome measurement. Employing a technique that we dub \emph{cluster resummation}, we isolate one large connected component while resumming all disconnected components before taking norms. Standard cluster counting then yields exponential suppression. The proof for minimum-weight decoding uses similar resummation techniques and also yields exponential suppression, albeit with a larger error bound and a smaller guaranteed threshold. Our channel expansion method can be readily applied to general local completely positive trace-preserving (CPTP) errors. These advances fill a longstanding gap in the rigorous understanding of coherent error thresholds for general qLDPC codes and provides an analytical framework for extending the analysis to more realistic QEC and FTQC settings.

The remainder of the paper is organized as follows.  Sec.~\ref{sec:setting}
formally defines the error model, recovery criterion, and states the main theorems.  Sec.~\ref{sec:strategy}  introduces the proof strategy of optimal recovery threshold in the AQEC framework. Sec.~\ref{sec:discussion} concludes with  discussion and outlook. In the Appendix, we present the comprehensive proof and technical details, as well as  a parallel proof for the minimum-weight decoder setting.

\section{Setting and main results}
\label{sec:setting}

\subsection{Codes and the check connectivity graph}

Let $\cC_n$ be a qLDPC code with
parameters $\qparams{n}{k}{d}$, and let $\cG_n$ be stabilizer generator set of $\cC_n$. Suppose that each check acts on at most $w$ qubits and each
physical qubit participates in at most $\ell$ checks, where
$w$ and $\ell$ are constants independent of $n$.

Define the \emph{check connectivity graph} $\Gamma_n$ on the
physical qubits, with two qubits adjacent if they share a check
in $\cG_n$. Its maximum degree therefore satisfies \(\deg(\Gamma_n)\le \ell(w-1)\).

Throughout this paper, we consider the following \emph{local coherent error} model. 
Fix positive integers $D$ and $r$, independent of $n$. 
Each error realization is described by a depth-$D$ unitary circuit $U_n(\eta)=U_{n,D}(\eta)\cdots U_{n,1}(\eta)$. Each layer has the form $U_{n,t}(\eta)=\prod_{\alpha\in A_{n,t}}e^{-\mathrm{i}\eta h_{n,t,\alpha}}$ and consists of gates acting on pairwise disjoint sets of at most $r$ qubits. The generators $h_{n,t,\alpha}$ are arbitrary Hermitian operators satisfying $\norm{h_{n,t,\alpha}}\le1$ and $1\le \lvert\operatorname{supp}(h_{n,t,\alpha})\rvert\le r$. Gates in different layers need not commute.
We stress that the generators need not be Pauli operators, so the model goes beyond simple Pauli rotations.  Moreover, the circuit architecture may vary with $n$,
with $D$ and $r$ fixed, and our results hold uniformly over all allowed noise realizations.

In particular, this definition naturally captures coherent gate miscalibration. A noisy logical gate implemented by a constant-depth circuit can be decomposed into the ideal logical gate followed by a local coherent error of constant depth and bounded support. More broadly, this decomposition extends our code capacity analysis to coherent imperfections in such logical gate implementations, thereby linking the present result to the study of fault tolerance under coherent noise. A more rigorous statement and proof can be found in Appendix~\ref{app:common-setting}.

\subsection{Syndrome measurements and recovery criteria}
 
Unlike stochastic errors, a coherent error \(U\) acts unitarily on the code state without leaking logical information to an environment.
If $U$ were known explicitly  and exactly,
one could in principle recover the original
code state by simply applying \(U^\dagger\).
In practice, however, we do not generally have full knowledge of $U$, and implementing its inverse would itself be subject to control errors.
Therefore, the relevant setting is standard error correction in which we measure syndromes after the coherent error, rather than being able to directly reverse $U$. The measurement can leak logical information to the environment, making recovery nontrivial. We consider optimal recovery as an fundamental benchmark of what can be recovered in principle, and minimum-weight decoding as a concrete recovery rule determined by the code. The former yields a tighter logical error bound and a higher threshold, while the latter has more direct practical
relevance.

Let $\{\Pi_s:s\in\Sigma_n\}$ be the projectors onto the syndrome sectors labeled by $s$. The measurement channel is {$\cM_n(Y)=\sum_{s\in\Sigma_n}\ket{s}\!\bra{s}_S\otimes \Pi_sY\Pi_s$}.
The recovery map has access to both the classical syndrome register \(S\)
and the postmeasurement physical system.
Syndrome measurements remove coherence between different syndrome
sectors while preserving coherence within each sector.

We now define the logical error measure for both the optimal recovery and the minimum-weight decoding.
For a given fully specified local coherent error $U$, let
$\cN_{n,U,\eta}=\cM_n\circ\Ad_{U_n(\eta)}\circ\cE_n$. The
corresponding \emph{optimal recovery logical error} is defined by
\begin{equation}
    \eps_n^\star(U,\eta)
    =
    \inf_{\cR}
    \frac12\norm{\cR\circ\cN_{n,U,\eta}-\id}_{\diamond},
    \label{eq:optimal-error}
\end{equation}
The infimum is over CPTP recovery maps from the syndrome register and
the postmeasurement physical system to the logical space.
The order of quantifiers in Eq.~\eqref{eq:optimal-error} is essential: the recovery is optimized for each fixed noise realization before a supremum over all allowed local coherent errors is taken. 
Let $\cR^{\rm MW}$ denote the recovery map which, for each observed syndrome $s$, applies a minimum-weight Pauli correction with the same syndrome.
This
\emph{minimum-weight decoding logical error} is defined by
\begin{equation}
    \eps_n^{\rm MW}(U,\eta)
    =
    \frac12\norm{
    \cR^{\rm MW}\circ\cN_{n,U,\eta}-\id
    }_{\diamond},
    \label{eq:MW-error}
\end{equation}
This recovery depends only on the code and requires no knowledge of
the noise model.

\subsection{Coherent error threshold results}

We now state our main  threshold results.
\begin{theorem}[Optimal recovery threshold, informal]
\label{thm:informal}
Suppose a stabilizer qLDPC code family has distance
$d_n=\Omega(\log n)$.
There exist constants $\eta_{\rm th}>0$, $C>0$, and $c>0$ such that
\begin{equation}
    \sup_U \eps_n^\star(U,\eta)
    \le C \sqrt{n} e^{-c d_n}
    \label{eq:informal-optimal}
\end{equation}
for every $n$ and every $|\eta|<\eta_{\rm th}$.  The supremum is taken over all realizations
of local coherent error while fixing $\eta$. 
\end{theorem}

\begin{theorem}[Minimum-weight decoding threshold, informal]
\label{thm:MW-informal}
Suppose a stabilizer qLDPC code family has distance
$d_n=\Omega(\log n)$.
There exist constants $\eta^{\rm MW}_{\rm th}>0$, $C_{\rm MW}>0$, and $c_{\rm MW}>0$
such that
\begin{equation}
    \sup_U \eps_n^{\rm MW}(U,\eta)
    \le C_{\rm MW} n e^{-c_{\rm MW} d_n}
    \label{eq:informal-mwd}
\end{equation}
for every $n$ and every $|\eta|<\eta^{\rm MW}_{\rm th}$.
The supremum is taken over all realizations of local coherent error
while fixing $\eta$.
\end{theorem}

A proof sketch of the optimal recovery threshold is given in Sec.~\ref{sec:strategy}, while full proofs of both theorems are given in Appendices~\ref{app:optimal} and \ref{app:minimum-weight}.
Our bounds contain only polynomial prefactors in \(n\) and are independent of \(k\), while retaining exponential suppression in the code distance. To our knowledge, these are the most general bounds on the logical error in diamond distance under coherent noise with this scaling. This behavior closely parallels logical error suppression under local stochastic noise~\cite{Gottesman2014} and substantially improves on previous bounds containing factors exponential in \(n\)~\cite{HuangDohertyFlammia2019,ChristandlFawziGoswami2025}.
Consequently, our results apply to qLDPC families with both vanishing or nonvanishing code rates as long as the distance is not worse than logarithmic. Relevant examples include topological codes~\cite{Kitaev2003,DennisEtAl2002,BombinMartinDelgado2006Color, BombinMartinDelgado2006Rate,BreuckmannTerhal2016}, bivariate bicycle codes~\cite{BravyiEtAl2024BB}, hypergraph product codes~\cite{TillichZemor2014}, quantum expander codes~\cite{LeverrierTillichZemor2015}, asymptotically good qLDPC codes~\cite{PanteleevKalachev2022,LeverrierZemor2022,DinurEtAl2023}, and almost-good qLTCs~\cite{DinurLinVidick2024}. For any such family, our theorems establish a nonzero threshold against arbitrary local coherent errors.
This broad scope of validity encompasses practically relevant constructions offering features such as geometric locality, local testability, or low implementation overhead, substantially expanding the applicability of the previous threshold guarantee relying on constant rate and linear distance~\cite{ChristandlFawziGoswami2025}. Many of these constructions play central roles in current proposals for QEC and FTQC architectures~\cite{GoogleQuantumAI2023,BluvsteinEtAl2024,XuEtAl2024, BravyiEtAl2024BB}.

Compared with the minimum-weight decoding bound in Eq.~\eqref{eq:informal-mwd}, the optimal recovery bound in Eq.~\eqref{eq:informal-optimal} has a smaller polynomial prefactor, \(\sqrt{n}\) instead of \(n\). Our estimates also give \(c>c_{\rm MW}\) and \(\eta_{\rm th}>\eta^{\rm MW}_{\rm th}\). Both thresholds are lower bounds on the optimal code capacity threshold.

 We express our main bounds in terms of diamond distance, which is particularly well suited to QEC because it controls the worst case logical error even when the encoded block is entangled with an arbitrary reference system~\cite{KretschmannSchlingemannWerner2008,IversonPreskill2020}. Diamond distance does not increase when ideal channels are applied before or after the noisy channel, and errors accumulate at most additively under successive composition. These properties allow a code capacity bound to be incorporated into a larger FTQC analysis~\cite{IversonPreskill2020}. Composing a code capacity bound with surrounding ideal channels therefore does not increase the error. The diamond distance characterization is especially important for coherent errors because average infidelity can substantially underestimate their worst-case effects. For small unitary errors, the diamond distance scales as the square root of the infidelity, whereas for stochastic Pauli noise the diamond distance scales linearly with the infidelity~\cite{KuengEtAl2016,IversonPreskill2020}. This difference in scaling is relevant both theoretically and experimentally. Experiments have observed substantial reductions in diamond distance after noise randomization, with little change in average fidelity~\cite{WareEtAl2021}.

It is worth emphasizing that locality here refers to bounded support rather than geometric proximity. 
Each term in the coherent error acts on at most \(r\) qubits, but when a spatial geometry is present, those qubits can be arbitrarily far apart.
Thus, even for geometrically local topological codes, the noise need not respect the underlying geometry. Our results therefore cover coherent errors of bounded support with arbitrarily long spatial range.

The logical error bound holds uniformly over all allowed coherent errors. Subject only to the stated noise model constraints, the gate arrangement gates, supports, and generators may otherwise be chosen arbitrarily. No specific microscopic noise model or assumptions of independence or phase randomness are required. The results therefore apply even when the detailed error structure depends on the device and is difficult to characterize experimentally.

\section{Proof sketch}
\label{sec:strategy}
Here we sketch the proof of the optimal recovery threshold, which captures the key insights underlying our approach. The proof for minimum-weight decoding follows a similar strategy while requiring additional combinatorial arguments. Complete proofs of both results are given in Appendices~\ref{app:optimal} and \ref{app:minimum-weight}. Our argument roughly proceeds in three steps.
\begin{enumerate}
\item
We bound the optimal recovery error in terms of how strongly
the syndrome distribution depends on the logical input. 

\item
Code distance ensures that configurations with only small connected components contribute no logical dependence.

\item
We develop a new resummation technique to bound the total contribution of configurations containing large connected components. Counting the possible large components then yields the threshold.
\end{enumerate}
For the remainder of this sketch, we fix a certain noise realization and suppress the dependence on $n$ and $\eta$ until the final estimate.

We first characterize recoverability using a complementary channel $\cQ_U$, which captures logical information leaked into the environment. In our setting, $\cQ_U$ outputs only the syndrome distribution. If this distribution is independent of the logical input, exact recovery is possible.
Quantitatively, \cite[Theorem~3]{KretschmannSchlingemannWerner2008} gives
\begin{equation}
 \eps^\star(\cN_U)
 \le
 \sqrt{
 \norm{\cQ_U\circ(\id-\cD_A)}_\diamond
 },
 \label{eq:information-disturbance}
\end{equation}
where $\cD_A(X)=\Tr(X)I_A/\dim A$.
The norm measures how strongly the syndrome distribution depends on the logical input and could be viewed as an AQEC measure.
The construction of $\cQ_U$ is given in Appendix~\ref{app:syndrome-recovery}.

Expand the noise channel using the local {error} maps
$F_\alpha\coloneqq \Ad_{u_\alpha}-\id$:
\begin{equation}
 \Ad_U=\sum_{B\subseteq\Lambda}F_B.
 \label{eq:error-map}
\end{equation}
Here $\Lambda$ labels the circuit locations, $u_\alpha$ is
the gate at location $\alpha$, and $F_B$ is the composition
of the maps $F_\alpha$ for $\alpha\in B$ in circuit order,
with $F_\varnothing=\id$. For every nonempty \(B\), \(F_B\) is a Hermiticity-preserving, trace-annihilating map, which we treat directly as an elementary error configuration rather than expanding in the Pauli basis. This formulation extends naturally to more general noise channels. In particular, the main theorem could be generalized to local CPTP errors, see Appendix \ref{app:local-cptp}.

We organize each error configuration $B$ using an \emph{interaction graph} $H$ on $\Lambda$, in which two locations are adjacent if their supports overlap or both intersect the support of a common stabilizer check.
The connected components of the induced graph $H[B]$ define the {error} components
of $B$. Since each gate acts on at most $r$ qubits, a component
with fewer than $m\coloneqq \lceil d/r\rceil$ locations involves fewer
than $d$ qubits.
If all components are small, the configuration contributes no logical dependence, even when its total support size exceeds $d$. This is because no stabilizer check meets two distinct components, so the code distance argument applies to each component separately.
Let $\cT_{\mathrm{small}}$ denote the sum of the corresponding maps
$F_B$. Hence it suffices to bound
$\norm{\Ad_U-\cT_{\mathrm{small}}}_\diamond$.

For a certain large component $C$, the expansion still includes
all {error} components disconnected from $C$.
At each available location, the identity term has diamond norm one and the {error} term has diamond norm at most $a\coloneqq 2\sin|\eta|$ for $|\eta|\le\pi/2$.
If we bound these terms separately, each available location contributes a factor of $1+a$, producing an overall bound that can grow exponentially with the circuit size. We avoid this growth with a new resummation technique that sums over the disconnected errors before taking norms.
Specifically, let $\nu_m(B)$ be the number of connected components of $H[B]$
containing at least $m$ locations.
We introduce the interpolation {for each subset of circuit locations $W\subseteq \Lambda$}
\begin{equation}
\cT_W(t)
\coloneqq 
\sum_{B\subseteq W}
t^{\nu_m(B)}F_B,
\qquad
0\leq t\leq1,
 \label{eq:sketch-interpolation}
\end{equation} 
with the convention $0^0=1$.
Each large component contributes one factor of $t$, so
\(\cT_\Lambda(0)=\cT_{\mathrm{small}},\ \cT_\Lambda(1)=\Ad_U.\)
The quantity we need to control is therefore \(\norm{\cT_\Lambda(1)-\cT_\Lambda(0)}_\diamond\).

\begin{figure}[t]
    \centering
    \includegraphics[width=\columnwidth]{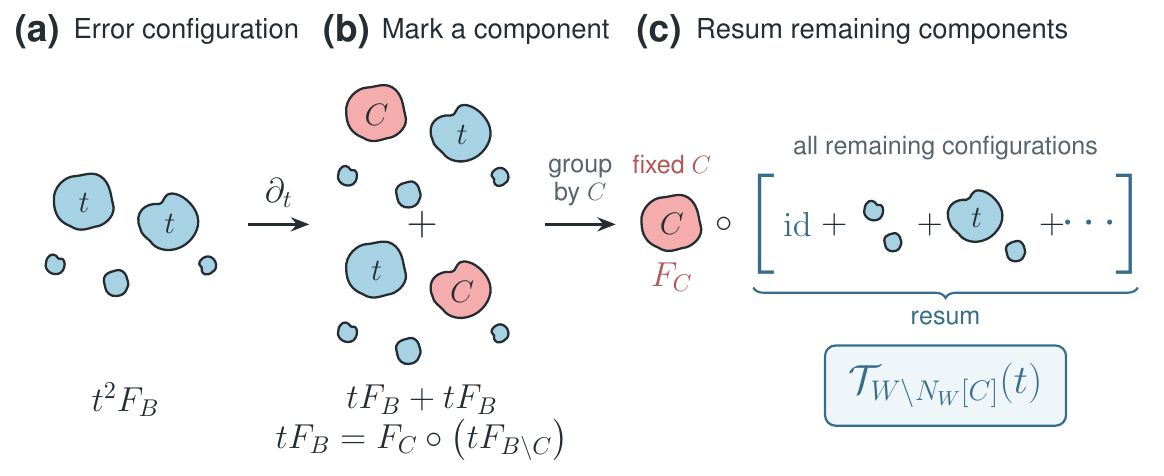}
    \caption{ Schematic illustration of cluster resummation. (a) A configuration in Eq. \eqref{eq:sketch-interpolation} that contains two large components and several small components. Each large component provides a factor $t$. (b) Differentiating with respect to $t$ is equivalent to marking each large component $C$ of $B$ in turn and removing its associated factor of $t$. (c) The terms are grouped by the marked component $C$, and $F_C$ is factored out. Summing the remaining terms, with their associated powers of $t$, over all configurations $B\setminus C\subseteq W\setminus N_W[C]$ gives $\cT_{W\setminus N_W[C]}(t)$, yielding Eq.~\eqref{eq:sketch-derivative-identity}.
    }
    \label{fig:cluster-resummation}
\end{figure}
{
The key observation is that differentiating Eq.~\eqref{eq:sketch-interpolation} isolates each large component of $W$ while exactly resumming other disconnected error configurations as another $\cT$ channel,
\begin{equation}
 \cT_W'(t)=
 \sum_{\substack{C\subseteq W: H[C]\ \text{\normalfont connected}\\|C|\ge m}}
 F_C\,\cT_{W\setminus N_W[C]}(t),
 \label{eq:sketch-derivative-identity}
\end{equation}
where $N_W[C]$ is the closed neighborhood of $C$ in $H[W]$.
We call this  \emph{cluster resummation}. This technique plays the role of the union bound in threshold proofs for local stochastic noise~\cite{Gottesman2014}, without introducing a factor exponential in the circuit size. The derivation of the cluster resummation is illustrated in Fig.~\ref{fig:cluster-resummation}, and a detailed proof is in Appendix \ref{app:cluster-resummation}.
}

\begin{theorem}[Resummation bound]
\label{thm:sketch-resummation}
For unitary noise circuit $U$ with interaction graph $H$,
suppose $\norm{F_\alpha}_\diamond\le a$ for every
$\alpha\in\Lambda$, and define
\begin{equation}
 \Xi \coloneqq 
 \sum_{\substack{
 C\subseteq\Lambda:\ H[C]\ \text{\normalfont connected}\\
 |C|\ge m}}
 a^{|C|}.
\end{equation}
If $\Xi<1$, then
\begin{equation}
 \begin{aligned}
 \norm{\Ad_U-\cT_{\mathrm{small}}}_\diamond
 \le \frac{\Xi}{1-\Xi}.
 \end{aligned}
 \label{eq:interpolation-bound}
\end{equation}
\end{theorem}
Intuitively, after integrating Eq.~\eqref{eq:sketch-derivative-identity}
and taking the norm, \(F_C\) contributes the factor \(a^{|C|}\),
\(\cT_{W\setminus N_W[C]}\) is bounded by a constant, and the sum over
\(C\) is controlled by standard cluster counting.
The proof is given in Appendix~\ref{app:ordered-expansion}.

Since $\cT_{\mathrm{small}}$ contributes no logical dependence,
Theorem~\ref{thm:sketch-resummation} and Eq.~\eqref{eq:information-disturbance} give $\eps^\star\le\sqrt{2\Xi/(1-\Xi)}$ for $\Xi<1$, the factor $2$ comes from $\norm{\id-\cD_A}_\diamond\le2$.
Together with $\eps^\star\le1$, this yields $\eps^\star\le2\sqrt{\Xi}$ for all $\Xi\ge0$. The counting estimate in Appendix~\ref{app:counting-constants}
gives $\Xi\le Dn\,q^m/[\chi(1-q)]$ for $q\coloneqq \chi a<1$,
where $\chi$ depends only on $w,\ell,D,r$.
Using $m\ge d_n/r$, we obtain
\begin{equation}
 \sup_U\eps_n^\star(U,\eta)
 \le
 2\sqrt{\frac{D}{\chi(1-q)}}\,
 \sqrt n\,q^{d_n/(2r)}.
 \label{eq:main-bound}
\end{equation}
If $d_n\ge c_d\log n$ for some $c_d>0$, choose $0<q_{\rm th}<e^{-r/c_d}$. Since $q=2\chi\sin|\eta|\to0$ as $\eta\to0$, there exists $\eta_{\rm th}>0$ such that $q\le q_{\rm th}$ whenever $|\eta|<\eta_{\rm th}$. The bound then decays exponentially in $d_n$. If $d_n/\log n\to\infty$, the bound vanishes for any fixed $0\le q<1$.

\section{Discussion}
\label{sec:discussion}
For highly general qLDPC codes under local coherent noise, we established constant code capacity thresholds below which the logical error decays exponentially with the code distance. Our results show that the sparse check structure of qLDPC codes together with logarithmic code distance already suffices to guarantee robustness against general local coherent error, thereby establishing a long-elusive general and rigorous understanding of coherent quantum error correction. 
A key technical advance enabling these results is cluster resummation, which allows us to bound coherent channel expansions without introducing an exponential dependence on system size. We apply the method to both optimal recovery and minimum-weight decoding, and to unitary coherent noise as well as general local CPTP noise. The minimum-weight decoding result establishes a threshold for a concrete decoder. An important future direction is to develop efficient decoders for coherent errors in specific qLDPC families and establish rigorous threshold guarantees.

Moreover, the present results establish code capacity thresholds  under perfect syndrome measurements. A natural next step is to consider circuit-level noise where syndrome extraction is itself noisy and repeated measurements must be decoded  based on the spacetime syndrome history. Establishing thresholds for coherent and more general local CPTP faults in this setting would help extend our analysis toward more practical qLDPC-based FTQC settings. A further challenge is to treat non-Markovian noise, where a persistent environment can induce correlations across space and time ~\cite{TerhalBurkard2005,AharonovKitaevPreskill2006}.  Continued exploration of these results and considerations is expected to provide an essential route toward a comprehensive theory of spacetime quantum fault tolerance.

\begin{acknowledgments}
Z.H., Y.Z., and Z.-W.L.\ are supported in part by NSFC under Grant No.~12475023, Dushi Program, and a startup funding from YMSC. Y.W.\ is supported by a startup funding from SIMIS.
\end{acknowledgments}

\bibliography{coherent}


\begin{appendix}

\section{Setting details}
\label{app:common-setting}

We use the code and noise model of Sec.~\ref{sec:setting}
and collect the notation needed in the appendices.

Let \(\cC_n\) be an \(\qparams{n}{k_n}{d_n}\) stabilizer
code with Pauli stabilizer generators \(\cG_n\).
Each generator has weight at most \(w\), and each qubit
belongs to at most \(\ell\) generator supports, where
\(w\) and \(\ell\) are independent of \(n\).
The logical and physical Hilbert spaces are
\(A_n=\mathbb C^{2^{k_n}}\) and
\(P_n=(\mathbb C^2)^{\otimes n}\).
Fix an encoding isometry \(V_n:A_n\to P_n\), with
code projector \(P_{\cC_n}=V_nV_n^\dagger\) and encoding
channel \(\cE_n=\Ad_{V_n}\), where
\(\Ad_K(X)=KXK^\dagger\).

The check connectivity graph \(\Gamma_n\) has the physical
qubits as vertices, with two qubits adjacent when they
occur in the support of a common generator. Thus
\begin{equation}
\deg(\Gamma_n)\le\kappa,
\qquad
\kappa\coloneqq \ell(w-1).
\label{eq:app-kappa}
\end{equation}

For a given code, we omit the subscript \(n\).
Fix positive integers \(D\) and \(r\), independent of \(n\).
A local coherent error is represented by the circuit
\begin{equation}
\begin{aligned}
U\coloneqq U_n(\eta)&=U_D\cdots U_1,\\
U_t&=\prod_{a\in\mathcal A_t}
e^{-\mathrm i\eta h_{t,a}},
\end{aligned}
\label{eq:app-circuit}
\end{equation}
where each \(h_{t,a}\) is Hermitian and satisfies
\(\norm{h_{t,a}}\le1\) and
\(1\le|\supp(h_{t,a})|\le r\).
Supports within each layer are pairwise disjoint.

{
This definition captures coherent gate miscalibration, which is described in the following proposition.
\begin{figure}[t]
    \centering
    \includegraphics[width=\columnwidth]{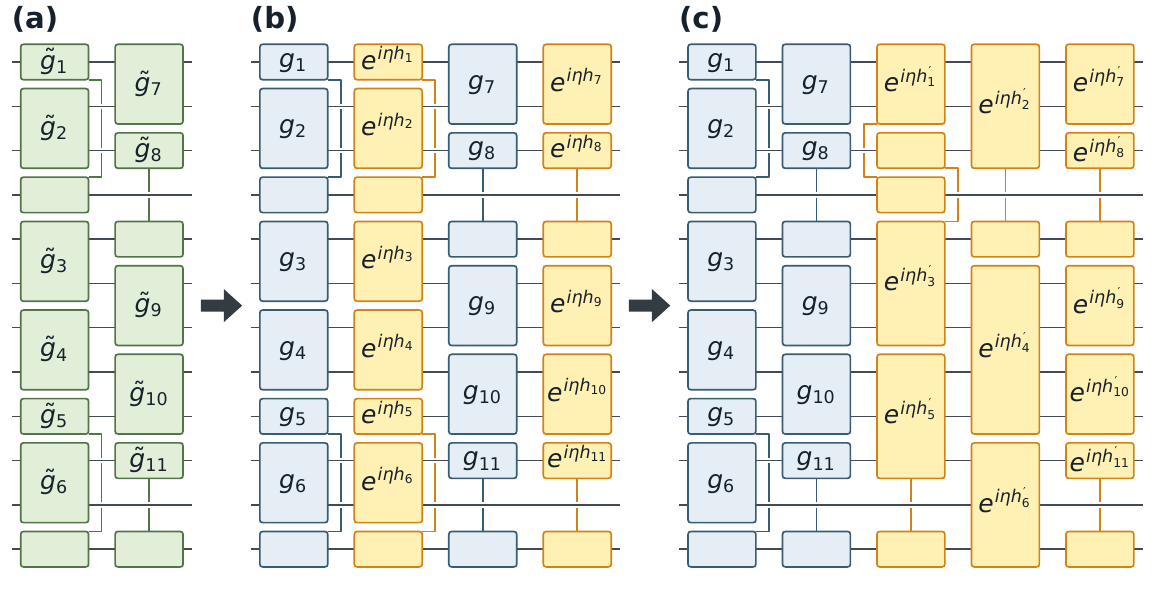}
    \caption{{
    Local coherent error arising from a constant-depth logical operation.
    (a) A two-layer noisy logical operation composed of bounded-support noisy physical gates \(\tilde g_i\).
    (b) Each \(\tilde g_i\) can be decomposed into an ideal physical gate \(g_i\) and a coherent error gate \(e^{i\eta h_i}\).
    (c) Propagating every error gate \(e^{i\eta h_i}\) through the physical circuit allows the decomposation of the full circuit into a two-layer ideal logical operation followed by a three-layer local coherent error.}
    }
    \label{fig:local-coherent-error}
\end{figure}

\begin{proposition}[Noisy constant-depth logical operations]
Suppose a logical operation is implemented by an ideal circuit
\begin{equation}
G=G_D\cdots G_1,
\qquad
G_t=\prod_{\alpha\in A_t}g_{t,\alpha},
\label{eq:ideal-logical-circuit}
\end{equation}
where \(D=O(1)\) and all physical gates have support size at most \(r=O(1)\),
with pairwise-disjoint supports within each layer. Let
\(\widetilde g_{t,\alpha}\) denote the corresponding noisy unitary physical gate
and $\widetilde G$ be the noisy circuit.
Then every noisy physical gate admits a decomposition
\begin{equation}
\begin{aligned}
&\widetilde g_{t,\alpha}
=u_{t,\alpha}g_{t,\alpha},
\qquad
u_{t,\alpha}=e^{-\mathrm i\eta h_{t,\alpha}}, \qquad h_{t,\alpha}=h_{t,\alpha}^\dagger,
\end{aligned}
\label{eq:gate-noise-decomposition}
\end{equation}
with $\norm{h_{t,\alpha}}\le1$ and for a common \(\eta\ge0\). After propagating all \(u_{t,\alpha}\) through the
remaining ideal gates, the noisy implementation can be written as
\begin{equation}
\widetilde G=UG,
\label{eq:propagated-noisy-logical-circuit}
\end{equation}
where \(U\) is a local coherent error in
Eq.~\eqref{eq:app-circuit}, with depth at most $D r^{2D-1}$ and support bound at most $r^D$.
\end{proposition}

\begin{proof}
The settings are illustrated in Fig.~\ref{fig:local-coherent-error}. Define
\begin{equation}
v_{t,\alpha}
\coloneqq \widetilde g_{t,\alpha}g_{t,\alpha}^\dagger.
\label{eq:relative-unitary-error}
\end{equation}
Since \(v_{t,\alpha}\) is unitary, there exists a Hermitian operator
\(K_{t,\alpha}\) on the same physical-gate support such that
\begin{equation}
v_{t,\alpha}
=e^{-\mathrm iK_{t,\alpha}},
\qquad
\norm{K_{t,\alpha}}\le\pi.
\label{eq:relative-error-logarithm}
\end{equation}
Set
\begin{equation}
\eta\coloneqq \max_{t,\alpha}\norm{K_{t,\alpha}},
\qquad
h_{t,\alpha}\coloneqq 
\begin{cases}
K_{t,\alpha}/\eta, & \eta>0,\\
0, & \eta=0.
\end{cases}
\label{eq:common-noise-strength}
\end{equation}
Then
\begin{equation}
\widetilde g_{t,\alpha}
=e^{-\mathrm i\eta h_{t,\alpha}}g_{t,\alpha},
\qquad
\norm{h_{t,\alpha}}\le1,
\label{eq:gate-error-factorization}
\end{equation}
which proves Eq.~\eqref{eq:gate-noise-decomposition}.

Because the supports within each layer are pairwise disjoint,
\begin{equation}
\widetilde G_t=U_tG_t,
\qquad
U_t\coloneqq \prod_{\alpha\in A_t}u_{t,\alpha}.
\label{eq:noisy-layer-decomposition}
\end{equation}
Define
\begin{equation}
S_t\coloneqq G_D\cdots G_{t+1},
\qquad
S_D\coloneqq I.
\label{eq:remaining-ideal-circuit}
\end{equation}
Successively moving the ideal gates to the right gives
\begin{equation}
\widetilde G
=\left(\prod_{t=D}^{1}S_tU_tS_t^\dagger\right)G,
\label{eq:propagation-identity}
\end{equation}
where the product is ordered with decreasing \(t\). For every \(t\),
\begin{equation}
S_tU_tS_t^\dagger
=\prod_{\alpha\in A_t}e^{-\mathrm i\eta h'_{t,\alpha}},
\qquad
h'_{t,\alpha}\coloneqq  S_th_{t,\alpha}S_t^\dagger,
\label{eq:propagated-generators}
\end{equation}
and unitary conjugation gives
\begin{equation}
\norm{h'_{t,\alpha}}
=\norm{h_{t,\alpha}}
\le1.
\label{eq:propagated-generator-norm}
\end{equation}

Conjugation through one layer of pairwise-disjoint \(r\)-supported gates maps an
operator supported on \(q\) qubits to one supported on at most \(rq\) qubits.
Therefore
\begin{equation}
\bigl|\supp(h'_{t,\alpha})\bigr|
\le r^{D-t+1}
\le r^D.
\label{eq:propagated-support-bound}
\end{equation}
For fixed \(t\), the original \(u_{t,\alpha}\) commute, hence
\begin{equation}
\left[
S_tu_{t,\alpha}S_t^\dagger,
S_tu_{t,\beta}S_t^\dagger
\right]=0
\qquad
(\alpha,\beta\in A_t).
\label{eq:propagated-errors-commute}
\end{equation}
Let \(k=D-t\). A backward light cone through \(k\) layers contains at most
\(r^k\) qubits. Since the original supports in \(A_t\) are pairwise disjoint,
the number of propagated supports from \(A_t\) intersecting a certain propagated
support is at most
\begin{equation}
r^{k+1}r^k=r^{2k+1}\le r^{2D-1}.
\label{eq:propagated-overlap-bound}
\end{equation}
Hence their intersection graph has maximum degree at most
\(r^{2D-1}-1\) and can be colored with at most \(r^{2D-1}\) colors. By Eq.~\eqref{eq:propagated-errors-commute}, the color classes may be
used as sublayers without changing the product. Applying this independently
to the \(D\) original layers gives a circuit representation of \(U\) with depth at
most
\begin{equation}
Dr^{2D-1}=O(1)
\label{eq:effective-noise-depth}
\end{equation}
and support size at most
\begin{equation}
r^D=O(1).
\label{eq:effective-noise-support}
\end{equation}
Thus \(U\) satisfies the conditions of Eq.~\eqref{eq:app-circuit}.
\end{proof}}

To expand the noise channel into local {error} contributions,
let \(\Lambda\) denote the set of gate locations in \(U\),
and let \(u_\alpha\) be the unitary at location \(\alpha\).
Define
\begin{equation}
F_\alpha\coloneqq \Ad_{u_\alpha}-\id.
\label{eq:app-error-map}
\end{equation}
For \(B\subseteq\Lambda\), let \(F_B\) be the composition
of the maps \(F_\alpha\), \(\alpha\in B\), in the order
in which the corresponding gates act in \(U\), with
\(F_\varnothing=\id\).
Expanding each local channel as
\(\Ad_{u_\alpha}=\id+F_\alpha\) gives
\[
\Ad_U=\sum_{B\subseteq\Lambda}F_B.
\]

After the coherent error, the stabilizer generators are
measured. Let \(\{\Pi_s:s\in\Sigma\}\) be their nonzero
joint eigenspace projectors. They satisfy
\begin{equation}
\Pi_s\Pi_t=\delta_{st}\Pi_s,
\qquad
\sum_{s\in\Sigma}\Pi_s=I_P.
\label{eq:app-syndrome-projectors}
\end{equation}
The measurement channel outputs the syndrome in a register
\(S\), together with the postmeasurement physical state:
\begin{equation}
\cM(Y)
=
\sum_{s\in\Sigma}
\ket{s}\!\bra{s}_S\otimes\Pi_sY\Pi_s.
\label{eq:app-retained-measurement}
\end{equation}
Combining encoding, noise, and measurement gives
\begin{equation}
\cN_U\coloneqq \cN_{n,U,\eta}
=
\cM\circ\Ad_U\circ\cE.
\label{eq:app-noisy-channel}
\end{equation}

\section{Optimal recovery threshold}
\label{app:optimal}
In this section, we establish the coherent error threshold for the optimal recovery setting.
The \emph{optimal recovery logical error} is defined by
\begin{equation}
\eps_n^\star(U,\eta)
=
\inf_{\cR }
\frac12
\norm{\cR\circ\cN_U-\id_A}_\diamond,
\label{eq:app-optimal-error}
\end{equation}
where the infimum is over completely positive
trace-preserving (CPTP) maps from \(S\otimes P\) to \(A\).
The recovery is optimized separately for each circuit \(U_n(\eta)\).

\subsection{Logical information leakage and optimal recovery}
\label{app:syndrome-recovery}

In this section, we use the dependence of the syndrome distribution on the logical
input to quantify information leakage and bound the optimal recovery error.

For each \(s\in\Sigma\), define
\[
K_s\coloneqq \Pi_sUV.
\]
The postmeasurement channel in Eq.~\eqref{eq:app-noisy-channel} can then
be written as
\begin{equation}
\cN_U(X)
=
\sum_{s\in\Sigma}
\ket{s}\!\bra{s}_S\otimes K_sXK_s^\dagger .
\label{eq:appendix-postmeasurement-channel}
\end{equation}

Let \(E\) be a Hilbert space with orthonormal basis
\(\{\ket{s}_E:s\in\Sigma\}\), and define
\begin{equation}
W:A\rightarrow S\otimes P\otimes E,
\quad
W
=
\sum_{s\in\Sigma}
\ket{s}_S\otimes K_s\otimes\ket{s}_E .
\label{eq:appendix-W}
\end{equation}
We have
\begin{align}
W^\dagger W
&=
\sum_{s\in\Sigma}K_s^\dagger K_s
\nonumber\\
&=
V^\dagger U^\dagger
\left(\sum_{s\in\Sigma}\Pi_s\right)
UV
=
I_A .
\label{eq:W-isometry}
\end{align}

Hence \(W\) is an isometry. For any operator \(X\) on \(A\),
\begin{equation}
WXW^\dagger=\sum_{s,t\in\Sigma}
\ket{s}\!\bra{t}_S
\otimes K_sXK_t^\dagger
\otimes \ket{s}\!\bra{t}_E.
\end{equation}
Tracing out \(E\) therefore gives
\begin{equation}
\Tr_E\!\left(WXW^\dagger\right)
=
\sum_{s\in\Sigma}
\ket{s}\!\bra{s}_S\otimes K_sXK_s^\dagger
=
\cN_U(X).
\label{eq:appendix-stinespring}
\end{equation}
Thus \(W\) is a Stinespring isometry for \(\cN_U\), with output
system \(S\otimes P\) and environment \(E\).

A complementary channel corresponding to \(W\) is
\begin{align}
\cQ_U(X)
&\coloneqq 
\Tr_{SP}\!\left(WXW^\dagger\right)
\nonumber\\
&=
\sum_{s\in\Sigma}
\Tr\!\left(K_sXK_s^\dagger\right)
\ket{s}\!\bra{s}_E
\nonumber\\
&=
\sum_{s\in\Sigma}
\Tr\!\left(K_s^\dagger K_sX\right)
\ket{s}\!\bra{s}_E .
\label{eq:appendix-complement}
\end{align}
For every logical state \(\rho\), the diagonal entries of
\(\cQ_U(\rho)\) are precisely the probabilities of the syndrome
outcomes. We now compare these probabilities for different logical
inputs.

Let
\begin{equation}
\tau_A\coloneqq \frac{I_A}{\dim A},
\qquad
\cD_A(X)\coloneqq \Tr(X)\tau_A,
\label{eq:logical-depolarizing-channel}
\end{equation}
where \(\cD_A\) is the completely depolarizing channel on \(A\).
For every logical state \(\rho\),
\begin{equation}
\bigl[\cQ_U\circ(\id-\cD_A)\bigr](\rho)
=
\cQ_U(\rho)-\cQ_U(\tau_A).
\label{eq:syndrome-input-dependence}
\end{equation}
Thus the syndrome distribution is independent of the logical input
exactly when
\begin{equation}
\cQ_U\circ(\id-\cD_A)=0.
\label{eq:syndrome-independence}
\end{equation}
More generally,
\(\norm{\cQ_U\circ(\id-\cD_A)}_\diamond\)
measures {the information leakage to the environment, and equivalently} the dependency of input logical states.

The {error} expansion of \(\cQ_U\) contains terms that are independent
of the logical input. The following elementary observation gives a
convenient characterization of such terms.

\begin{lemma}
\label{lem:input-independence}
Let \(L\) be a linear map with input system \(A\). Then
\begin{equation}
L\circ(\id-\cD_A)=0
\quad\Longleftrightarrow\quad
L(X)=\Tr(X)L(\tau_A)
\label{eq:input-independent-characterization}
\end{equation}
for every operator \(X\) on \(A\).
Moreover, this equality remains valid in the presence of an arbitrary
reference system: for every finite-dimensional \(R\),
\begin{equation}
(L\otimes\id_R)
\circ
\bigl[(\id-\cD_A)\otimes\id_R\bigr]
=
0.
\label{eq:input-independent-cancellation}
\end{equation}
\end{lemma}

\begin{proof}
The first condition is equivalent to \(L=L\circ\cD_A\). Hence
\begin{equation}
L(X)
=
L\bigl(\Tr(X)\tau_A\bigr)
=
\Tr(X)L(\tau_A).    
\end{equation}

The converse follows from the same calculation. The last statement
follows by tensoring the equality
\(L\circ(\id-\cD_A)=0\) with \(\id_R\).
\end{proof}

We now relate this measure of syndrome dependence to the optimal
recovery error.

\begin{lemma}
\label{lem:recovery-from-syndrome}
The optimal recovery error of \(\cN_U\) satisfies
\begin{equation}
\epsilon^\star(\cN_U)
\leq
\sqrt{
\norm{\cQ_U\circ(\id-\cD_A)}_\diamond
}.
\label{eq:recovery-from-syndrome}
\end{equation}
\end{lemma}

\begin{proof}
Since \(\cQ_U\) is a complementary channel of \(\cN_U\) and
\(\cQ_U\circ\cD_A\) is a constant channel, {applying Theorem~3 of Ref. \cite{KretschmannSchlingemannWerner2008} in  Schr\"odinger picture \cite[Theorem~3]{HaydenWinter2012}} gives
\begin{equation}
\frac{1}{4}
\left(
\inf_{\cR}
\norm{\cR\circ\cN_U-\id_A}_\diamond
\right)^2
\leq
\norm{\cQ_U-\cQ_U\circ\cD_A}_\diamond.    
\end{equation}
The right-hand side is
\(
\norm{\cQ_U\circ(\id-\cD_A)}_\diamond,
\)
while the left-hand side is
\(\epsilon^\star(\cN_U)^2\) by definition. Taking the square root
proves Eq.~\eqref{eq:recovery-from-syndrome}.
\end{proof}

It remains to bound
\(\norm{\cQ_U\circ(\id-\cD_A)}_\diamond\).
In the next appendix, we show that a {error} contribution \(L_B\)
satisfies
\begin{equation}
    L_B(X)=\Tr(X)\sigma_B
\end{equation}
whenever every connected component of its physical support contains
fewer than \(d\) qubits. Here \(\sigma_B\) may depend on \(B\), but
not on \(X\). Lemma~\ref{lem:input-independence} then gives
\begin{equation}
L_B\circ(\id-\cD_A)=0.    
\end{equation}
Consequently, only {error} configurations containing a connected
component of size at least \(d\) need to be bounded.

\subsection{Error clusters and logical leakage}
\label{app:subdistance-cancellation}

This appendix shows that {error} configurations consisting only of
small connected components do not contribute to the dependence of
the syndrome distribution on the logical input. We first relate
connected components in the interaction graph to connected components
of their physical support. We then use the code distance to show that
the corresponding syndrome contributions are independent of the
logical input.

The local coherent noise is defined by a depth-\(D\)
circuit
\[
U=U_D\cdots U_1,
\qquad
U_t=\prod_{a\in\mathcal A_t}e^{-i\eta h_{t,a}},
\]
where \(\mathcal A_t\) indexes the local unitary in layer
\(t\). Each \(h_{t,a}\) is supported on at most \(r\) physical
qubits, and the supports are pairwise disjoint within each layer.

We call the pair \(\alpha=(t,a)\), which identifies one local
unitary in the circuit, a location. The set of
locations is
\begin{equation}
\Lambda
\coloneqq 
\bigl\{(t,a):1\leq t\leq D,\ a\in\mathcal A_t\bigr\}.
\end{equation}
For \(\alpha=(t,a)\in\Lambda\), write
\[
h_\alpha\coloneqq h_{t,a},
\qquad
u_\alpha\coloneqq e^{-i\eta h_\alpha},
\qquad
S_\alpha\coloneqq \supp(h_\alpha).
\]

We now compare connectivity among coherent error locations with connectivity of their physical support. The check connectivity graph \(\Gamma\) has the physical qubits as vertices. The interaction graph \(H\) has vertex set \(\Lambda\), with two locations adjacent whenever their supports have distance at most one in \(\Gamma\).

For \(B\subseteq\Lambda\), define its physical support by 
\[
S_B\coloneqq \bigcup_{\alpha\in B}S_\alpha.
\]
The induced graph \(H[B]\) describes
connectivity among the locations in \(B\), while
\(\Gamma[S_B]\) describes connectivity within their physical support.
The following lemma bounds the size and degree of the interaction graph and relates its connected components to the connected components of their physical support.

\begin{lemma}
\label{lem:layered-geometry}
The coherent error locations and their supports satisfy the following:
\begin{enumerate}
    \item The number of locations     \(|\Lambda|\leq Dn.\)

    \item The maximum degree of the interaction graph \(H\) is at most \(D(\kappa+1)r-1.\)

    \item If \(C\subseteq\Lambda\) is connected in \(H\), then
    \(|S_C|\leq r|C|.\)

    \item For every \(B\subseteq\Lambda\) and every connected component \(T\) of \(\Gamma[S_B]\), there is a connected component \(C\) of \(H[B]\) such that \(T\subseteq S_C\), \(|T|\leq |S_C|\leq r|C|\).
\end{enumerate}
\end{lemma}

\begin{proof}
Within each layer, the supports are nonempty and pairwise disjoint,
so there are at most \(n\) locations. Summing over the \(D\) layers
proves item 1.

Let \(N_\Gamma[S_\alpha]\) denote the closed neighborhood
of \(S_\alpha\) in \(\Gamma\). Since
\(\deg(\Gamma)\leq\kappa\) and \(|S_\alpha|\leq r\),
\begin{equation}
|N_\Gamma [S_\alpha]|
\leq
(\kappa+1)|S_\alpha|
\leq
(\kappa+1)r.    
\end{equation}
In any layer, at most \(|N_\Gamma[S_\alpha]|\) pairwise disjoint supports can intersect \(N_\Gamma[S_\alpha]\). Hence at most \(D(\kappa+1)r\) locations are equal or adjacent to \(\alpha\) in
\(H\), which proves item 2 after excluding \(\alpha\) itself.

Item 3 follows directly from
\begin{equation}
    |S_C|
=
\left|\bigcup_{\alpha\in C}S_\alpha\right|
\leq
\sum_{\alpha\in C}|S_\alpha|
\leq
r|C|.
\end{equation}

For item 4, let
\begin{equation}
I_T\coloneqq \{\alpha\in B:S_\alpha\cap T\neq\varnothing\}.    
\end{equation}
This set is connected in \(H[B]\). Indeed, for
\(\alpha,\beta\in I_T\), choose
\(x\in S_\alpha\cap T\) and \(y\in S_\beta\cap T\), and let
\[
x=v_0,v_1,\ldots,v_p=y
\]
be a path in \(T\). For each \(i\), choose
\(\gamma_i\in B\) such that \(v_i\in S_{\gamma_i}\), with
\(\gamma_0=\alpha\) and \(\gamma_p=\beta\). Consecutive locations
\(\gamma_i\) and \(\gamma_{i+1}\) are either equal or adjacent in
\(H[B]\). Thus they form a walk from \(\alpha\) to \(\beta\).

Let \(C\) be the connected component of \(H[B]\) containing \(I_T\).
Every vertex of \(T\) lies in \(S_\alpha\) for some
\(\alpha\in I_T\), so \(T\subseteq S_C\). The final bound follows
from item 3.
\end{proof}

\subsection{Recoverability of small errors}

Set $m\coloneqq \left\lceil\frac{d}{r}\right\rceil$.
If every connected component of \(H[B]\) has fewer than \(m\)
locations, Lemma~\ref{lem:layered-geometry} implies that every
connected component \(T\) of \(\Gamma[S_B]\) satisfies
\begin{equation}
|T|\leq r(m-1)<d.    
\end{equation}

We first record the corresponding consequence for operators supported
on \(S_B\).

\begin{lemma}
\label{lem:app-scalar-compression}
Let \(R\) be a set of physical qubits such that every connected
component of \(\Gamma[R]\) has fewer than \(d\) vertices. For any
operators \(O_1,O_2\) supported on \(R\) and any syndrome
\(s\in\Sigma\), there is a scalar \(\lambda_s(O_1,O_2)\) such that
\begin{equation}
P_{\cC}O_2^\dagger\Pi_sO_1P_{\cC}
=
\lambda_s(O_1,O_2)P_{\cC}.
\label{eq:app-scalar-compression}
\end{equation}
\end{lemma}

\begin{proof}
Expand \(O_1\) and \(O_2\) in the Pauli basis. For a Pauli operator
\(E\), let \(s(E)\) denote its syndrome. Then
\[
\Pi_sEP_{\cC}
=
\begin{cases}
EP_{\cC}, & s(E)=s,\\
0, & s(E)\neq s.
\end{cases}
\]
Hence a term $P_{\cC}E_2^\dagger\Pi_sE_1P_{\cC}$
can be nonzero only if \(s(E_1)=s(E_2)=s\). For such a pair,
\(E_2^\dagger E_1\) has trivial syndrome and is supported on \(R\).
Every connected component of its support is contained in a connected
component of \(\Gamma[R]\), and therefore has fewer than \(d\)
vertices. By Lemma~1 of ~\cite{KovalevPryadko2013},
\begin{equation}
P_{\cC}E_2^\dagger E_1P_{\cC}
=
\lambda(E_1,E_2)P_{\cC}.    
\end{equation}

Summing over the Pauli expansions of \(O_1\) and \(O_2\) proves the
claim.
\end{proof}

We now apply Lemma~\ref{lem:app-scalar-compression} to the {error}
expansion. Recall that
\[
\Ad_U=\sum_{B\subseteq\Lambda}F_B.
\]
Let
\begin{equation}
\cT_{\mathrm{small}}
\coloneqq 
\sum_{\substack{B\subseteq\Lambda:\\
\text{every component of }H[B]\text{ has size }<m}}
F_B
\label{eq:app-small-map}
\end{equation}
denote the part of the expansion in which every connected component
has fewer than \(m\) locations.

Define the syndrome channel on the physical system by
\begin{equation}
\cQ_{\mathrm{syn}}(Y)
\coloneqq 
\sum_{s\in\Sigma}
\Tr(\Pi_sY)\ket{s}\!\bra{s}.
\label{eq:app-syndrome-channel}
\end{equation}
The complementary channel \(\cQ_U\) can be written as
\begin{equation}
\cQ_U
=
\cQ_{\mathrm{syn}}\circ\Ad_U\circ\Ad_V.
\label{eq:complementary-channel-factorization}
\end{equation}
Indeed,
\begin{equation}
\begin{aligned}
\cQ_{\mathrm{syn}}
\circ\Ad_U\circ\Ad_V(X)
&=
\sum_{s\in\Sigma}
\Tr\!\left(\Pi_sUVXV^\dagger U^\dagger\right)
\ket{s}\!\bra{s} \\
&=
\sum_{s\in\Sigma}
\Tr\!\left(K_s^\dagger K_sX\right)
\ket{s}\!\bra{s}
=
\cQ_U(X).
\end{aligned}
\end{equation}

\begin{lemma}
\label{lem:app-small-cancellation}
There is a diagonal operator \(\sigma_{\mathrm{small}}\) on the
syndrome register such that
\begin{equation}
\cQ_{\mathrm{syn}}
\circ\cT_{\mathrm{small}}
\circ\Ad_V(X)
=
\Tr(X)\sigma_{\mathrm{small}}.
\label{eq:app-small-trace-only}
\end{equation}
Consequently,
\begin{equation}
\cQ_{\mathrm{syn}}
\circ\cT_{\mathrm{small}}
\circ\Ad_V
\circ(\id-\cD_A)
=
0.
\label{eq:app-small-cancellation}
\end{equation}
\end{lemma}

\begin{proof}
For \(B\subseteq\Lambda\) and \(K\subseteq B\), let \(u_K\) denote
the product of the unitaries \(u_\alpha\), \(\alpha\in K\), in circuit
order, with \(u_\varnothing=I\). Expanding the {error} maps gives
\begin{equation}
F_B
=
\sum_{K\subseteq B}
(-1)^{|B|-|K|}\Ad_{u_K}.
\label{eq:explicit-FB}
\end{equation}

Consider a set \(B\) appearing in
Eq.~\eqref{eq:app-small-map}. Every connected component of
\(\Gamma[S_B]\) has fewer than \(d\) qubits. Since
\(\supp(u_K)\subseteq S_B\),
Lemma~\ref{lem:app-scalar-compression} gives
\[
V^\dagger u_K^\dagger\Pi_su_KV
=
\lambda_s(u_K,u_K)I_A
\]
for every \(s\in\Sigma\). Hence
\begin{align}
\cQ_{\mathrm{syn}}
\circ\Ad_{u_K}
\circ\Ad_V(X)
&=
\sum_{s\in\Sigma}
\Tr\!\left(
V^\dagger u_K^\dagger\Pi_su_KVX
\right)
\ket{s}\!\bra{s}
\nonumber\\
&=
\Tr(X)
\sum_{s\in\Sigma}
\lambda_s(u_K,u_K)\ket{s}\!\bra{s}.
\label{eq:app-small-unitary-term}
\end{align}
Substituting Eq.~\eqref{eq:explicit-FB} and summing over the sets
\(B\) in Eq.~\eqref{eq:app-small-map}, the coefficient of each
\(\ket{s}\!\bra{s}\) remains proportional to \(\Tr(X)\). This proves
Eq.~\eqref{eq:app-small-trace-only}.
Equation~\eqref{eq:app-small-cancellation} then follows from
Lemma~\ref{lem:input-independence}.
\end{proof}

Since
\[
\Ad_U-\cT_{\mathrm{small}}
=
\sum_{\substack{B\subseteq\Lambda:\\
H[B]\text{ has a component of size }\geq m}}
F_B,
\]
Eqs.~\eqref{eq:complementary-channel-factorization} and
\eqref{eq:app-small-cancellation} give
\begin{equation}
\begin{aligned}
\cQ_U\circ(\id-\cD_A)
&=
\cQ_{\mathrm{syn}}
\circ\Ad_U
\circ\Ad_V
\circ(\id-\cD_A) \\
&=
\cQ_{\mathrm{syn}}
\circ
\bigl(\Ad_U-\cT_{\mathrm{small}}\bigr)
\circ\Ad_V
\circ(\id-\cD_A).
\end{aligned}
\label{eq:large-component-reduction}
\end{equation}
Thus only {error} sets \(B\) for which \(H[B]\) contains a connected
component of at least \(m\) locations remain. The next appendix bounds
their total contribution.

\subsection{Bounding large cluster contributions}
\label{app:ordered-expansion}

Appendix~\ref{app:subdistance-cancellation} shows that only {error} configurations containing a large connected cluster can contribute to logical information leakage. We now bound their total contribution.

For a certain large connected cluster, there are exponentially many possible choices of {errors} elsewhere in the circuit. Bounding these choices separately would introduce an exponential dependence on the total number of circuit locations. We instead sum over all remaining {errors} before taking the norm.

We first record a bound on the strength of a single channel {error}.
\begin{lemma}
\label{lem:single-fault-norm}
If \(|\eta|\leq\pi/2\), then
\begin{equation}
\norm{F_\alpha}_\diamond
\leq
2\sin|\eta|.
\label{eq:appendix-sine}
\end{equation}
\end{lemma}

\begin{proof}
Since \(F_\alpha\) is the difference of two channels, the optimization
in its diamond norm may be restricted to pure joint input states. Let
\(\ket{\psi}_{PR}\) be a pure state of the physical system \(P\) and
an arbitrary reference system \(R\). Using the trace distance between
pure states, we obtain
\begin{align}
&\norm{
(F_\alpha\otimes\id_R)
\bigl(\ket{\psi}\!\bra{\psi}\bigr)
}_1
\nonumber\\
&\quad=
\norm{
(u_\alpha\otimes I_R)\ket{\psi}\!\bra{\psi}
(u_\alpha^\dagger\otimes I_R)
-
\ket{\psi}\!\bra{\psi}
}_1
\nonumber\\
&\quad=
2\sqrt{
1-
\left|
\bra{\psi}u_\alpha\otimes I_R\ket{\psi}
\right|^2
}.
\label{eq:pure-state-distance}
\end{align}

Let
\(\rho_P\coloneqq \Tr_R(\ket{\psi}\!\bra{\psi})\).
Since \(h_\alpha\) is Hermitian and
\(\norm{h_\alpha}\leq 1\), it has a spectral decomposition
\begin{equation}
h_\alpha
=
\sum_j\lambda_j\ket{j}\!\bra{j},
\qquad
\lambda_j\in[-1,1].
\label{eq:local-generator-spectrum}
\end{equation}
Writing
\(p_j\coloneqq \bra{j}\rho_P\ket{j}\), we have
\begin{equation}
\bra{\psi}u_\alpha\otimes I_R\ket{\psi}
=
\Tr(\rho_Pu_\alpha)
=
\sum_jp_je^{-\mathrm{i}\eta\lambda_j}.
\label{eq:local-unitary-overlap}
\end{equation}
For \(|\eta|\leq\pi/2\), Eq.~\eqref{eq:local-unitary-overlap}
satisfies
\begin{align}
\operatorname{Re}\Tr(\rho_Pu_\alpha)
=
\sum_jp_j\cos(\eta\lambda_j)
\nonumber\geq
\cos|\eta|,
\label{eq:local-unitary-real-part}
\end{align}
where we used \(|\lambda_j|\leq1\). Therefore,
\begin{equation}
\left|
\bra{\psi}u_\alpha\otimes I_R\ket{\psi}
\right|
\geq
\cos|\eta|.
\label{eq:local-unitary-overlap-bound}
\end{equation}
Substituting Eq.~\eqref{eq:local-unitary-overlap-bound} into
Eq.~\eqref{eq:pure-state-distance} gives
\begin{equation}
\norm{
(F_\alpha\otimes\id_R)
\bigl(\ket{\psi}\!\bra{\psi}\bigr)
}_1
\leq
2\sin|\eta|.
\label{eq:single-fault-state-bound}
\end{equation}
The bound is uniform over the pure input state and the reference
system. Hence
\begin{equation}
\norm{F_\alpha}_\diamond
\leq
2\sin|\eta|,
\end{equation}
as claimed.
\end{proof}

\subsection{Expansion over large clusters}
\label{app:cluster-resummation}

For \(B\subseteq\Lambda\), call a connected component of \(H[B]\)
large if it contains at least \(m\) locations, and let
\(\nu_m(B)\) denote the number of such components. For every
\(W\subseteq\Lambda\), define
\begin{equation}
\cT_W(t)
\coloneqq 
\sum_{B\subseteq W}
t^{\nu_m(B)}F_B,
\qquad
0\leq t\leq1,
\label{eq:app-interpolation}
\end{equation}
with the convention \(0^0=1\). 

For \(W=\Lambda\), the endpoints are
\begin{equation}
\cT_\Lambda(1)=\Ad_U,
\qquad
\cT_\Lambda(0)=\cT_{\mathrm{small}}.
\label{eq:app-interpolation-endpoints}
\end{equation}
Hence
\begin{equation}
\Ad_U-\cT_{\mathrm{small}}
=
\int_0^1\cT_\Lambda'(t)\,\mathrm{d}t.
\label{eq:app-interpolation-integral}
\end{equation}

We next derive a useful expression for the derivative
\(\cT_W'(t)\). For \(C\subseteq W\), let \(N_W[C]\) denote the
closed neighborhood of \(C\) in the induced graph \(H[W]\).

\begin{lemma}
\label{lem:cluster-identity}
For every \(W\subseteq\Lambda\) and \(0\leq t\leq1\),
\begin{equation}
\cT_W'(t)
=
\sum_{\substack{C\subseteq W:\\
H[C]\text{ connected},\ |C|\geq m}}
F_C\circ
\cT_{W\setminus N_W[C]}(t).
\label{eq:appendix-derivative}
\end{equation}
\end{lemma}

\begin{proof}
Differentiating Eq.~\eqref{eq:app-interpolation} gives
\begin{equation}
\cT_W'(t)
=
\sum_{B\subseteq W}
\;
\sum_{\substack{
C\text{ a connected component of }H[B]\\
|C|\geq m}}
t^{\nu_m(B)-1}F_B.
\label{eq:raw-derivative}
\end{equation}
Here the inner sum accounts for the factor \(\nu_m(B)\) produced by
differentiation.

For each pair \((B,C)\) in
Eq.~\eqref{eq:raw-derivative}, set \(R\coloneqq B\setminus C\). Since \(C\)
is a connected component of \(H[B]\),
\[
R\subseteq W\setminus N_W[C].
\]
Conversely, any connected \(C\subseteq W\) with \(|C|\geq m\),
together with any \(R\subseteq W\setminus N_W[C]\), gives
\(B=C\cup R\) with \(C\) as a large connected component of \(H[B]\).
Thus the sum could be indexed by \(C\) and \(R\).

There are no edges between \(C\) and \(R\), so the corresponding
{error} maps commute across the two sets. Moreover, removing \(C\)
removes exactly one large component. Hence
\begin{equation}
F_B=F_C\circ F_R,
\qquad
\nu_m(B)=1+\nu_m(R).
\label{eq:marked-factorization}
\end{equation}
Reindexing Eq.~\eqref{eq:raw-derivative} and summing first over \(R\)
now gives
\begin{align}
\cT_W'(t)
&=
\sum_{\substack{C\subseteq W:\\
H[C]\text{ connected},\ |C|\geq m}}
F_C\circ
\sum_{R\subseteq W\setminus N_W[C]}
t^{\nu_m(R)}F_R
\nonumber\\
&=
\sum_{\substack{C\subseteq W:\\
H[C]\text{ connected},\ |C|\geq m}}
F_C\circ
\cT_{W\setminus N_W[C]}(t),
\end{align}
which proves Eq.~\eqref{eq:appendix-derivative}.
\end{proof}

\subsection{Bounding the resummed expansion}

Equation~\eqref{eq:appendix-derivative} expresses
\(\cT'_W(t)\) as a sum of terms
\(F_C\circ\cT_{W\setminus N_W[C]}(t)\).
The single-{error} bound \(\norm{F_\alpha}_\diamond\le a\),
with \(a=2\sin|\eta|\), and submultiplicativity give
\begin{equation}
\norm{F_C}_\diamond
\le \prod_{\alpha\in C}\norm{F_\alpha}_\diamond
\le a^{|C|}.
\end{equation}
To bound the second factor, we use the fact that
\(W\setminus N_W[C]\subseteq\Lambda\) and seek a uniform
bound on \(\norm{\cT_W(t)}_\diamond\) over all
\(W\subseteq\Lambda\) and \(0\le t\le1\).
Accordingly, define
\begin{equation}
    M \coloneqq  \max_{\substack{W\subseteq\Lambda\\0\le t\le1}} \norm{\cT_W(t)}_\diamond.
\end{equation}
For each \(W\), the finite sum defining \(\cT_W(t)\) depends continuously on \(t\), and hence so does its diamond norm. This norm attains a finite maximum on the closed interval \([0,1]\). Since there are only finitely many subsets \(W\subseteq\Lambda\), \(M\) is finite.

Recall the total activity of large connected sets,
\[
\Xi =
\sum_{\substack{C\subseteq\Lambda:\\
H[C]\ \text{\normalfont connected},\ |C|\ge m}}
a^{|C|}.
\]
Using the bound on \(F_C\) and the definition of \(M\),
Eq.~\eqref{eq:appendix-derivative} gives, for every
\(W\subseteq\Lambda\) and \(0\le t\le1\),
\begin{equation}
\norm{\cT'_W(t)}_\diamond
\le
\sum_{\substack{C\subseteq W:\\
H[C]\ \text{\normalfont connected},\ |C|\ge m}}
a^{|C|}M
\le \Xi M.
\end{equation}

At \(t=1\), \(\cT_W(1)\) is the unitary channel corresponding to the circuit obtained from \(U\) by replacing every gate \(u_\alpha\) at a location \(\alpha\notin W\) with the identity. Hence \(\norm{\cT_W(1)}_\diamond=1\).
Integrating from \(t\) to \(1\) and using the derivative
bound above gives
\begin{align}
\norm{\cT_W(t)}_\diamond
&\le
\norm{\cT_W(1)}_\diamond
+\int_t^1\norm{\cT'_W(s)}_\diamond\,ds
\nonumber\\
&\le 1+(1-t)\Xi M.
\end{align}
Taking the maximum over \(W\subseteq\Lambda\) and
\(0\le t\le1\) yields \(M\le1+\Xi M\).
For \(\Xi<1\), this implies
\begin{equation}
M\le\frac{1}{1-\Xi}.
\label{eq:MR-bound}
\end{equation}

We now apply the same derivative bound with \(W=\Lambda\).
Using the endpoint identities
\(\cT_\Lambda(1)=\Ad_U\) and
\(\cT_\Lambda(0)=\cT_{\rm small}\), we obtain
\begin{align}
\norm{\Ad_U-\cT_{\rm small}}_\diamond
&\le
\int_0^1\norm{\cT'_\Lambda(t)}_\diamond\,dt
\nonumber\\
&\le \Xi M
\le \frac{\Xi}{1-\Xi}.
\label{eq:app-ordered-bound}
\end{align}

We now translate Eq.~\eqref{eq:app-ordered-bound} into
a bound on the logical error.
Lemma~\ref{lem:app-small-cancellation} shows that the
contribution of \(\cT_{\rm small}\) to the syndrome output
is independent of the logical input. In particular,
\[
\cQ_{\rm syn}\circ\cT_{\rm small}\circ\Ad_V
\circ(\id-\cD_A)=0.
\]
Since
\(\cQ_U=\cQ_{\rm syn}\circ\Ad_U\circ\Ad_V\),
we can therefore write
\begin{equation}
\cQ_U\circ(\id-\cD_A)
=
\cQ_{\rm syn}\circ(\Ad_U-\cT_{\rm small})
\circ\Ad_V\circ(\id-\cD_A).    
\end{equation}

The channels \(\cQ_{\rm syn}\) and \(\Ad_V\) have
diamond norm one, while
\(\norm{\id-\cD_A}_\diamond\le2\).
By submultiplicativity and
Eq.~\eqref{eq:app-ordered-bound}, for \(\Xi<1\) we obtain
\begin{align}
\norm{\cQ_U\circ(\id-\cD_A)}_\diamond
&\le
\norm{\Ad_U-\cT_{\rm small}}_\diamond
\norm{\id-\cD_A}_\diamond
\nonumber\\
&\le \frac{2\Xi}{1-\Xi}.
\label{eq:app-syndrome-tail}
\end{align}

Lemma~\ref{lem:recovery-from-syndrome} bounds the optimal
recovery error by the square root of this norm:
\[
\eps_n^\star(U,\eta)
\le
\sqrt{\norm{\cQ_U\circ(\id-\cD_A)}_\diamond}.
\]
Combining this relation with
Eq.~\eqref{eq:app-syndrome-tail} and
\(\eps_n^\star(U,\eta)\le1\) gives
\begin{equation}
\eps_n^\star(U,\eta)
\le
\min\left\{
1,\sqrt{\frac{2\Xi}{1-\Xi}}
\right\},
\qquad \Xi<1.
\label{eq:app-finite-code-bound}
\end{equation}
The next appendix bounds \(\Xi\) by counting connected
sets of at least \(m\) locations in the interaction graph.

\subsection{Coherent error threshold}
\label{app:counting-constants}

Appendix~\ref{app:ordered-expansion} bounds the logical error in terms
of the total contribution \(\Xi\) of large {error} clusters. We now
estimate \(\Xi\) and complete the proof of the coherent error threshold.

Let \(N=|\Lambda|\), and let \(N_j\) denote the number of sets
\(C\subseteq\Lambda\) with \(|C|=j\) for which \(H[C]\) is
connected. {We bound $N_j$ using the standard cluster counting lemma.

\begin{lemma}[Lemma~5, Ref. \cite{AliferisGottesmanPreskill2008}]
\label{lem:AGP-cluster-counting}
Let \(H\) be a graph with \(N\) vertices and maximum degree
at most \(\Delta\), set
\(\chi=\max\{1,\mathrm e\Delta\}\), and let
$J_t$ be any fixed set of $t$ vertices.  Denote by
$M_\Delta(j,J_t)$ the number of connected vertex sets
$C \subseteq V(H)$ with $|C|=j\ge t$ and $J_t\subseteq C$.
Then
\begin{equation}
    M_\Delta(j,J_t)
    \le
    \mathrm e^{\,t-1}
    \chi^{j-t}.
    \label{eq:AGP-cluster-count}
\end{equation}
\end{lemma}}

\begin{lemma}
\label{lem:connected-set-count}
Let \(H\) be a graph with \(N\) vertices and maximum degree
at most \(\Delta\), and set
\(\chi=\max\{1,\mathrm e\Delta\}\).
\(N_j\) satisfies
\begin{equation}
N_j\le N\chi^{j-1},
\qquad j\ge1.
\label{eq:connected-set-count}
\end{equation}
\end{lemma}
{
\begin{proof}
For a certain vertex $v$, Lemma~\ref{lem:AGP-cluster-counting}
with $t=1$ gives
\begin{equation}
    M_\Delta(j,\{v\})
    \le
    (\mathrm e\Delta)^{j-1}
    \le
    \chi^{j-1}.
\end{equation}
Summing over the $N$ possible choices of $v$ counts every connected
$j$-vertex set at least once.  Therefore
\begin{equation}
    N_j
    \le
    \sum_{v\in V(H)}M_\Delta(j,\{v\})
    \le
    N\chi^{j-1},
\end{equation}
which proves Eq.~\eqref{eq:connected-set-count}.
\end{proof}}

For the interaction graph, Lemma~\ref{lem:layered-geometry}
gives \(N=|\Lambda|\le Dn\) and
\begin{equation}
\deg(H)\le\Delta,
\qquad
\Delta\coloneqq D(\kappa+1)r-1.
\label{eq:app-conflict-degree}
\end{equation}
Set
\begin{equation}
    \chi\coloneqq \max\{1,\mathrm e\Delta\},
\qquad
q\coloneqq \chi a=2\chi\sin|\eta|.
\end{equation}

For \(|\eta|\le\pi/2\) and \(q<1\), grouping the terms
in \(\Xi\) by the size of the connected set gives
\begin{align}
\Xi
&=\sum_{j=m}^{N}N_j a^j
\nonumber\\
&\le N\sum_{j=m}^{\infty}\chi^{j-1}a^j
=\frac{N}{\chi}\sum_{j=m}^{\infty}q^j
\nonumber\\
&=\frac{N}{\chi}\frac{q^m}{1-q}
\le\frac{Dn}{\chi}\frac{q^m}{1-q}.
\label{eq:appendix-xi-tail}
\end{align}

We next substitute this estimate into the recovery bound.
When \(\Xi\le1/2\),
Eq.~\eqref{eq:app-finite-code-bound} gives
\begin{equation}
\eps_n^\star(U,\eta)
\le\sqrt{\frac{2\Xi}{1-\Xi}}
\le2\sqrt{\Xi}.
\label{eq:appendix-general-small-tail}
\end{equation}
When \(\Xi>1/2\), the bound
\(\eps_n^\star(U,\eta)\le1<2\sqrt{\Xi}\)
gives the same estimate.
Thus \(\eps_n^\star(U,\eta)\le2\sqrt{\Xi}\)
holds in both cases.

The preceding bounds hold for every allowed local coherent error \(U\).
Combining them and using \(m\ge d_n/r\) and \(0\le q<1\),
we obtain
\begin{align}
\sup_U\eps_n^\star(U,\eta)
&\le
2\sqrt{\frac{Dn}{\chi(1-q)}}\,q^{m/2}
\nonumber\\
&\le
2\sqrt{\frac{D}{\chi(1-q)}}\,
\sqrt n\,q^{d_n/(2r)}.
\label{eq:appendix-general-bound}
\end{align}

This estimate gives vanishing recovery error as \(n\to\infty\) for sufficiently small \(|\eta|\).
{
\begin{theorem}[Optimal recovery threshold]
\label{thm:main}
Fix positive integers $w,\ell,D,r$.  Suppose the
stabilizer qLDPC family above
satisfies
\begin{equation}
    d_n = \Omega(\log n)
    \label{eq:log-distance}
\end{equation}
for every $n$.  Define
\begin{equation}
 \begin{aligned}
& \chi\coloneqq \max\{1,\mathrm e(D[\ell(w-1)+1]r-1)\},\\
 &q(\eta)\coloneqq 2\chi\sin|\eta|.
 \end{aligned}
 \label{eq:q-eta-definition}
\end{equation}
For every $|\eta|\le\pi/2$, there exists a constant $q_{\rm th}$, depending only on the stated
code and circuit parameters, such that for every $n$, whenever $|\eta|< \eta_{\rm th}$, we have,
\begin{equation}
    \sup_U  \eps_n^\star(U,\eta) \le2\sqrt{\frac{D}{\chi(1-q)}} \sqrt n\,q^{d_n/(2r)}.
    \label{eq:appendix-main-bound}
\end{equation}
and
\begin{equation}
    \lim_{n\rightarrow\infty} \sup_U  \eps_n^\star(U,\eta) = 0.
    \label{eq:appendix-threshold-limit}
\end{equation}  The supremum is over all realizations
of local coherent error and obeying the stated
bounded-support and layer conditions while fixing $\eta$.  
Specifically, the threshold value can be chosen as 
\begin{equation}
\eta_{\rm th}=
\begin{cases}
\displaystyle
\arcsin\!\left(\frac{e^{-r/c_d}}{2\chi}\right),
& d_n\ge c_d\log n,\\[3mm]
\displaystyle
\arcsin\!\left(\frac{1}{2\chi}\right),
& d_n=\omega(\log n).
\end{cases}
\label{eq:eta-threshold}
\end{equation}
\end{theorem}}
\begin{proof}
For \(|\eta|\le\eta_{\rm th}\), we have
\(q=2\chi\sin|\eta|\le q_0\).
Equation~\eqref{eq:appendix-general-bound} therefore gives
\begin{equation}
\sup_{|\eta|\le\eta_{\rm th}}\sup_U
\eps_n^\star(U,\eta)
\le
2\sqrt{\frac{D}{\chi(1-q_0)}}\,
\sqrt n\,q_0^{d_n/(2r)}.    
\end{equation}

For the first statement, \(d_n\ge c_d\log n\) implies
\(n\le\mathrm e^{d_n/c_d}\). Hence
\begin{equation}
\sqrt n\,q_0^{d_n/(2r)}
\le
\exp\!\left[
-\frac{d_n}{2}
\left(\frac{\log(1/q_0)}{r}-\frac{1}{c_d}\right)
\right]
=e^{-cd_n}.
\label{eq:appendix-tail-mechanism}
\end{equation}
The condition \(q_0<\mathrm e^{-r/c_d}\) ensures \(c>0\).
Combining the preceding estimates and using
\(d_n\ge c_d\log n\), we obtain
\[
\sup_{|\eta|\le\eta_{\rm th}}\sup_U
\eps_n^\star(U,\eta)
\le Ce^{-cd_n}
\le Cn^{-cc_d}
\longrightarrow0.
\]
For the second statement, \(d_n/\log n\to\infty\)
and \(q_0<1\) imply
\begin{equation}
\log\!\left(\sqrt n\,q_0^{d_n/(2r)}\right)
=
\frac12\log n
-\frac{d_n}{2r}\log(1/q_0)
\longrightarrow-\infty.
\label{eq:appendix-superlog-tail}
\end{equation}
Thus \(\sqrt n\,q_0^{d_n/(2r)}\to0\), proving
Eq.~\eqref{eq:appendix-threshold-limit}.
\end{proof}

{
\subsection{Local CPTP error threshold}
\label{app:local-cptp}

The above argument also applies to general \emph{local CPTP error}.
Keep the code, supports $S_\alpha$, and layer conditions of
Appendix~\ref{app:common-setting}, and replace each local unitary
channel by a CPTP map $\Phi_\alpha$ on the same physical qubits.
The resulting noise circuit is
\begin{equation}
 \Phi=\Phi_D\circ\cdots\circ\Phi_1,
 \qquad
 \Phi_t=\bigotimes_{a\in\mathcal A_t}\Phi_{t,a},
 \label{eq:cptp-circuit}
\end{equation}
with the identity on untouched qubits understood. Each support has
size at most $r$ and need not be geometrically connected. Define
the local noise strength by
\begin{equation}
 \frac12\norm{\Phi_\alpha-\id_{S_\alpha}}_\diamond\le p,
 \qquad \alpha\in\Lambda,\quad 0\le p\le1.
 \label{eq:cptp-strength}
\end{equation}
Set $\cN_\Phi\coloneqq \cM\circ\Phi\circ\cE$, and define
$\eps_n^\star(\Phi,p)$ by Eq.~\eqref{eq:app-optimal-error}
with $\cN_U$ replaced by $\cN_\Phi$. Notice that unlike coherent error, the environment of CPTP error is no longer isomorphic to the syndrome space.  
Recovery applies to the syndrome
and postmeasurement physical state, but has no access to the noise environments.

\begin{theorem}[Local CPTP error threshold]
\label{thm:local-cptp}
Fix positive integers $w,\ell,D,r$, and suppose the stabilizer
qLDPC family satisfies $d_n=\Omega(\log n)$.
Let $\chi$ be as in Theorem~\ref{thm:main}, and set
\begin{equation}
 q(p)\coloneqq 2\chi\sqrt{2p}.
 \label{eq:cptp-q}
\end{equation}
For every $0\le p<p_{\rm th}$ and every $n$,
with $q=q(p)$,
\begin{equation}
 \sup_\Phi\eps_n^\star(\Phi,p)
 \le 2\sqrt{\frac{D}{\chi(1-q)}}\,
       \sqrt n\,q^{d_n/(2r)},
 \label{eq:cptp-main-bound}
\end{equation}
and $\lim_{n\to\infty}\sup_\Phi\eps_n^\star(\Phi,p)=0$.
The supremum is over all local CPTP noise realizations satisfying
Eqs.~\eqref{eq:cptp-circuit} and \eqref{eq:cptp-strength},
with the stated support and layer conditions. The threshold can
be chosen as
\begin{equation}
 p_{\rm th}=
 \begin{cases}
 \displaystyle\frac{e^{-2r/c_d}}{8\chi^2},
       & d_n\ge c_d\log n,\\[2mm]
 \displaystyle\frac{1}{8\chi^2},
       & d_n=\omega(\log n),
 \end{cases}
 \label{eq:cptp-threshold}
\end{equation}
where the second choice gives the larger threshold under the
stronger distance assumption.
\end{theorem}
\begin{proof}
We give the replacements needed in the AQEC proof.
Fix a code and omit the subscript $n$, writing
$P=(\mathbb C^2)^{\otimes n}$ for the physical Hilbert space.
For the given noise channel $\Phi$ in
Eq.~\eqref{eq:cptp-circuit}, the channel to be recovered is
$\cN_\Phi\coloneqq \cM\circ\Phi\circ\cE$.
To apply the AQEC recovery bound, we first construct a
complementary channel $\cQ_\Phi$ of $\cN_\Phi$.

By continuity of Stinespring dilations
\cite[Theorem~1]{KretschmannSchlingemannWerner2008}, each
$\Phi_\alpha$ has an isometric dilation $W_\alpha$ into the
physical support and an environment $E_\alpha$ such that
\begin{equation}
 \norm{W_\alpha-J_\alpha}\le\sqrt{2p},
 \qquad
 J_\alpha\ket{\psi}=\ket{\psi}\otimes\ket{0_\alpha}.
 \label{eq:cptp-local-dilation}
\end{equation}
Let $\widetilde E=\bigotimes_{\alpha\in\Lambda}E_\alpha$.
Composing the local isometries in the original circuit order,
with each acting trivially on previously introduced environments,
gives an isometry $W_\Phi:P\to P\otimes\widetilde E$ satisfying
\[
 \Phi=\Tr_{\widetilde E}\circ\Ad_{W_\Phi}.
\]
Combining this noise dilation with the measurement dilation
of Appendix~\ref{app:syndrome-recovery} gives
\begin{equation}
 \cQ_\Phi
 =
 (\cQ_{\mathrm{syn}}\otimes\id_{\widetilde E})
 \circ\Ad_{W_\Phi}\circ\cE.
 \label{eq:cptp-complement}
\end{equation}
Its output includes both the copied syndrome register $E$
and the noise environment $\widetilde E$.
For coherent noise, a
copied syndrome space alone gives the environment.
For general CPTP noise, the environment is different and must be specified.

We now expand the dilation appearing in
Eq.~\eqref{eq:cptp-complement}.
To place all local maps on the common Hilbert space
$P\otimes\widetilde E$, prepare the environments in
$\ket{0_{\widetilde E}}=\bigotimes_\alpha\ket{0_\alpha}$
and replace the local error map of
Eq.~\eqref{eq:app-error-map} by
\begin{equation}
 \begin{aligned}
 F_\alpha\coloneqq 
   (\Ad_{W_\alpha}-\Ad_{J_\alpha})\circ\Tr_{E_\alpha}, \qquad
 \norm{F_\alpha}_\diamond\le2\sqrt{2p}.
 \end{aligned}
 \label{eq:cptp-error-map}
\end{equation}
The partial trace removes the input environment $E_\alpha$,
while the isometries produce its output.
The norm bound follows by expanding the difference of the
two isometric channels and using
$\norm{AXB}_1\le\norm{A}\norm{X}_1\norm{B}$,
also with an arbitrary reference system.
Using the same ordered products $F_B$ of the modified error maps,
Eq.~\eqref{eq:cptp-complement} becomes
\begin{equation}
 \cQ_\Phi(X)
 =
 \sum_{B\subseteq\Lambda}
 (\cQ_{\mathrm{syn}}\otimes\id_{\widetilde E})
 \left[
 F_B\!\left(
 \cE(X)\otimes
 \ket{0_{\widetilde E}}\!\bra{0_{\widetilde E}}
 \right)
 \right].
 \label{eq:cptp-complement-expansion}
\end{equation}
Indeed, immediately before applying the channel $F_\alpha$, its environment
is still in $\ket{0_\alpha}$ and uncorrelated with the other
systems. Thus $\Ad_{J_\alpha}\circ\Tr_{E_\alpha}$ acts as the
identity on that input, and $\id+F_\alpha$ implements the
local dilation $\Ad_{W_\alpha}$.
Expanding these factors in circuit order gives the sum over
$B$ in Eq.~\eqref{eq:cptp-complement-expansion}.

In the definitions of $\cT_{\mathrm{small}}$ and $\cT_W(t)$,
replace each $F_B$ by the map
$Y\mapsto F_B(Y\otimes
\ket{0_{\widetilde E}}\!\bra{0_{\widetilde E}})$.
The same argument shows that $\cT_W(1)$ is the CPTP dilation
of the subcircuit on $W$, with unused environments left in
their initial states. Hence
$\norm{\cT_W(1)}_\diamond=1$ and
$\cT_\Lambda(1)=\Ad_{W_\Phi}$.
Here, \(\cT_W(t)\) is defined with all environment registers initialized in \(\ket{0_{\widetilde E}}\). Hence \(\cT_W(1)\) is CPTP, even though \(\id+F_\alpha\), viewed as a channel on arbitrary inputs of \(P\otimes\widetilde E\), need not be CPTP, unlike the coherent error case.

For a small-component set \(B\), we proceed as in the coherent error case, except that the noise environment is retained rather than traced out. Fix two basis states \(\ket{e},\ket{e'}\) of \(\widetilde E\). For a certain syndrome \(s\), the coefficient of \(\ket{e}\!\bra{e'}\) in the environment output has the form $\Tr\!\left[
V^\dagger O_{e'}^\dagger \Pi_s O_e V\,X
\right]$,
where \(O_e\) and \(O_{e'}\) are physical operators supported on \(S_B\). Since \(B\) has only small connected components, Lemma~\ref{lem:app-scalar-compression} gives
$$
V^\dagger O_{e'}^\dagger \Pi_s O_e V
=\lambda_{s,e,e'} I_A.
$$
The above coefficient is therefore \(\lambda_{s,e,e'}\Tr(X)\). Since this holds for every \(e,e'\), the entire environment output in each syndrome sector depends on \(X\) only through \(\Tr(X)\). Hence Eq.~\eqref{eq:app-small-cancellation} remains valid with \(\cQ_{\mathrm{syn}}\) replaced by \(\cQ_{\mathrm{syn}}\otimes\id_{\widetilde E}\).

The interaction graph and $m=\lceil d/r\rceil$ are unchanged.
Maps belonging to distinct components act on disjoint physical
supports and distinct environments, so
Lemma~\ref{lem:cluster-identity} applies to the modified
interpolation without further changes.
The derivative estimate and integration in
Appendix~\ref{app:ordered-expansion} use only this identity,
the bound $\norm{F_\alpha}_\diamond\le a$,
and the CPTP endpoints.
They therefore give Eq.~\eqref{eq:app-ordered-bound}
with $a=2\sqrt{2p}$ and $\Ad_U$ replaced by $\cT_\Lambda(1)$.

Applying Lemma~\ref{lem:recovery-from-syndrome} to the full
complement \eqref{eq:cptp-complement} gives
$\eps_n^\star(\Phi,p)\le2\sqrt{\Xi}$ by the same argument.
The counting estimate \eqref{eq:appendix-xi-tail}, with
$q=2\chi\sqrt{2p}$, then yields Eq.~\eqref{eq:cptp-main-bound}.
Finally, the logarithmic and superlogarithmic distance arguments
in Theorem~\ref{thm:main} apply with this $q$:
$q<e^{-r/c_d}$ and $q<1$, respectively, give the two choices
in Eq.~\eqref{eq:cptp-threshold} and the stated limit.
\end{proof}
}

\section{Minimum-weight decoding threshold}
\label{app:minimum-weight}

In this section, we establish a coherent error threshold for {minimum-weight Pauli decoding}. We adopt the code and coherent noise model of Appendix~\ref{app:common-setting} and assume perfect syndrome measurements. The correction depends only on the measured syndrome.

We follow the same three steps as in the proof of the threshold for optimal recovery: identifying configurations with no nontrivial logical action, separating the contribution of a large connected component through an interpolation, and bounding all such contributions by cluster resummation and counting. The main difference is that, under minimum-weight decoding, at least half of the qubits in each component lie in the error support. The contribution of a component therefore decays with its size, which gives a direct bound on the logical error after recovery.

\subsection{Minimum-weight decoder and channel expansion}
\label{sec:MW-channel-expansion}
For each valid syndrome \(s\in\Sigma\), fix a minimum-weight Pauli correction
\begin{equation}
 C_s\in\operatorname*{argmin}_{C:\,s(C)=s}\operatorname{wt}(C).
 \label{eq:MW-decoder-choice}
\end{equation}
Since \(C_sP_{\cC}C_s^\dagger=\Pi_s\), applying \(C_s^\dagger\)
after observing \(s\) returns the physical state to the code space.
We then decode with \(V^\dagger\) and denote the resulting recovery
by \(\cR^{\rm MW}\). We extend it arbitrarily to a CPTP map from
\(S\otimes P\) to \(A\). We define the \emph{minimum-weight decoding logical error} by
\begin{equation}
 \eps_n^{\rm MW}(U,\eta)
 \coloneqq \frac12\norm{\cR^{\rm MW}\circ\cN_U-\id_A}_\diamond.
 \label{eq:MW-error-app}
\end{equation}
Since \(\cR^{\rm MW}\) has already been specified, this definition
contains no minimization over recovery maps.  

For \(s\in\Sigma\), define
\begin{equation}
 \cR_s(Y)\coloneqq V^\dagger C_s^\dagger\Pi_sY\Pi_sC_sV.
 \label{eq:MW-recovery-branch}
\end{equation}
Then \(\sum_{s\in\Sigma}\cR_s=\cR^{\rm MW}\circ\cM\).
Using \(\Ad_U=\sum_{B\subseteq\Lambda}F_B\), we obtain
\begin{equation}
 \cR^{\rm MW}\circ\cN_U
 =\sum_{B\subseteq\Lambda}\sum_{s\in\Sigma}
   \cR_s\circ F_B\circ\cE.
 \label{eq:MW-syndrome-expansion}
\end{equation}

In the proof for optimal recovery, it was enough to consider the
connected components of the error support \(S_B\). For minimum-weight
decoding, we must also include the support of the correction, since
we bound the logical error after applying \(C_s\). Thus, for each
\((B,s)\), define
\begin{equation}
 S_{B,s}\coloneqq S_B\cup\supp C_s.
 \label{eq:MW-configuration-support}
\end{equation}
We then consider the connected components of
\(\Gamma[S_{B,s}]\).

\subsection{Recoverability of small errors}
\label{sec:MW-small-errors}

\begin{lemma}
\label{lem:MW-small-components}
For every \(B\subseteq\Lambda\) and \(s\in\Sigma\), the following
statements hold.
\begin{enumerate}
\item If every connected component of \(\Gamma[S_{B,s}]\) has fewer
than \(d\) vertices, then \(\cR_s\circ F_B\circ\cE\) is a scalar
multiple of \(\id_A\).

\item If \(\cR_s\circ F_B\circ\cE\ne0\), then every connected component
\(K\) of \(\Gamma[S_{B,s}]\) satisfies
\begin{equation}
 |K\cap S_B|\ge\frac{|K|}{2}.
 \label{eq:MW-half-coverage}
\end{equation}
\end{enumerate}
\end{lemma}

\begin{proof}
Expand \(F_B\) in the Pauli basis. Since \(F_B\) is supported on
\(S_B\), every term has the form
\(Y\mapsto E_1YE_2^\dagger\), where \(E_1\) and \(E_2\) are supported
on \(S_B\). The syndrome projections leave only terms with
\(s(E_1)=s(E_2)=s\).

Suppose that every connected component of \(\Gamma[S_{B,s}]\) has
fewer than \(d\) vertices. For each such term,
Lemma~\ref{lem:app-scalar-compression}, applied with
\(R=S_{B,s}\), \(O_1=E_i\), and \(O_2=C_s\) for \(i=1,2\), gives
\begin{equation}
 V^\dagger C_s^\dagger E_iV
 =\lambda_s(E_i,C_s)I_A,
 \qquad i=1,2.
 \label{eq:MW-small-residual}
\end{equation}
Thus every term in \(\cR_s\circ F_B\circ\cE\) is proportional to
\(\id_A\), and so is their sum. This proves the first statement.

For the second statement, suppose that
\(\cR_s\circ F_B\circ\cE\ne0\). Then the Pauli expansion contains
at least one term \(Y\mapsto E_1YE_2^\dagger\) with
\(s(E_1)=s(E_2)=s\). Fix \(E=E_1\). The residual operator
\(C_s^\dagger E\) has zero syndrome and is supported on \(S_{B,s}\).

Let \(K\) be a connected component of \(\Gamma[S_{B,s}]\).
By the definition of \(\Gamma\), any two qubits in the support of
the same stabilizer generator are adjacent. Hence the support of a stabilizer generator intersects at most one connected component of \(\Gamma[S_{B,s}]\). Since
\(C_s^\dagger E\) has zero syndrome, its restriction to \(K\) also
has zero syndrome. Therefore,
\(C_s[(C_s^\dagger E)|_K]\) is another Pauli correction with syndrome \(s\). Up to an overall
phase, it agrees with \(E\) on \(K\) and with \(C_s\) outside \(K\).
By the minimality of \(C_s\),
\begin{align}
 \operatorname{wt}(C_s)
 &\le
 \operatorname{wt}\bigl(C_s[(C_s^\dagger E)|_K]\bigr)
 \nonumber\\
 &=\operatorname{wt}(C_s|_{K^c})
   +\operatorname{wt}(E|_K).
 \label{eq:MW-component-comparison}
\end{align}
Since
\begin{equation}
    \operatorname{wt}(C_s)
 =
 \operatorname{wt}(C_s|_K)
 +\operatorname{wt}(C_s|_{K^c}),
\end{equation}
we have
\begin{equation}
    \operatorname{wt}(C_s|_K)
 \le\operatorname{wt}(E|_K)
 \le |K\cap S_B|.
\end{equation}
Finally, \(K\subseteq S_B\cup\supp C_s\) implies
\begin{equation}
 |K|
 \le |K\cap S_B|+\operatorname{wt}(C_s|_K)
 \le2|K\cap S_B|.
 \label{eq:MW-half-coverage-proof}
\end{equation}
This proves the second statement.
\end{proof}

Thus, a configuration can have a nontrivial logical action only if
\(\Gamma[S_{B,s}]\) contains a connected component of at least \(d\)
qubits. For every nonzero contribution, at least half of the qubits
in each component lie in \(S_B\).
\subsection{Expansion over large clusters}
\label{sec:MW-interpolation}
For \(B\subseteq\Lambda\) and \(s\in\Sigma\), call a connected
component of \(\Gamma[S_{B,s}]\) large if it contains at least \(d\)
qubits, and let \(\mathcal C_d(B,s)\) denote the set of large
components. As in Eq.~\eqref{eq:app-interpolation}, define
\begin{equation}
 \cT^{\rm MW}(t)
 \coloneqq 
 \sum_{B\subseteq\Lambda}\sum_{s\in\Sigma}
 t^{|\mathcal C_d(B,s)|}
 \cR_s\circ F_B\circ\cE,
 \qquad 0\le t\le1.
 \label{eq:MW-interpolation}
\end{equation}
We use the convention \(0^0=1\).

The endpoints are
\begin{equation}
 \cT^{\rm MW}(1)=\cR^{\rm MW}\circ\cN_U,\quad \cT^{\rm MW}(0)=c_\Lambda\id_A,
 \label{eq:MW-interpolation-endpoints}
\end{equation}
where the second identity follows from
Lemma~\ref{lem:MW-small-components} for some scalar \(c_\Lambda\).
The map \(\cT^{\rm MW}(0)\) need not preserve trace, so
\(c_\Lambda\) need not equal one. Hence
\begin{equation}
 \cR^{\rm MW}\circ\cN_U-c_\Lambda\id_A
 =
 \int_0^1
 \frac{\mathrm d}{\mathrm dt}\cT^{\rm MW}(t)\,\mathrm dt.
 \label{eq:MW-interpolation-integral}
\end{equation}

Differentiating Eq.~\eqref{eq:MW-interpolation} gives
\begin{equation}
 \begin{aligned}
 \frac{\mathrm d}{\mathrm dt}\cT^{\rm MW}(t)
 ={}&\sum_{B\subseteq\Lambda}\sum_{s\in\Sigma}
       \sum_{K\in\mathcal C_d(B,s)}
       t^{|\mathcal C_d(B,s)|-1}\\
 &\hspace{18mm}\times\cR_s\circ F_B\circ\cE.
 \end{aligned}
 \label{eq:MW-raw-derivative}
\end{equation}
In the proof for optimal recovery, for each large component \(C\),
the remaining errors sum to
\(\cT_{W\setminus N_W[C]}(t)\). This decomposition does not hold
for minimum-weight decoding because \(C_s\) depends on the total syndrome. We therefore group the terms by their large components
without separating the sum over \(s\).

Let \(\mathcal C_d\) denote all connected subsets \(K\) of the
vertices of \(\Gamma\) with \(|K|\ge d\). A collection
\(\mathcal J\subseteq\mathcal C_d\) is called compatible if its sets
are pairwise disjoint and no edge of \(\Gamma\) joins two distinct
sets in \(\mathcal J\). For every nonempty compatible
\(\mathcal J\), define
\begin{equation}
 \mathcal W_{\mathcal J}
 \coloneqq 
 \sum_{B\subseteq\Lambda}
 \sum_{\substack{
 s\in\Sigma:\,
 \mathcal J\subseteq\mathcal C_d(B,s)}}
 \cR_s\circ F_B\circ\cE.
 \label{eq:MW-joint-channel}
\end{equation}
The condition
\(\mathcal J\subseteq\mathcal C_d(B,s)\) means that every
\(K\in\mathcal J\) is a large connected component of
\(\Gamma[S_{B,s}]\).

For \(K\in\mathcal C_d(B,s)\), use
\(t=1+(t-1)\) for each
\(K'\in\mathcal C_d(B,s)\setminus\{K\}\). Then
\begin{align}
 t^{|\mathcal C_d(B,s)|-1}
 &=
 \prod_{K'\in\mathcal C_d(B,s)\setminus\{K\}}
 [1+(t-1)]
 \nonumber\\
 &=
 \sum_{\mathcal J\subseteq
       \mathcal C_d(B,s)\setminus\{K\}}
 (t-1)^{|\mathcal J|}.
 \label{eq:MW-t-expansion}
\end{align}
For fixed \(K\) and \(\mathcal J\), the sum over \(B\) and \(s\)
satisfying
\begin{equation}
    \{K\}\cup\mathcal J\subseteq\mathcal C_d(B,s)
\end{equation}
is \(\mathcal W_{\{K\}\cup\mathcal J}\). Reordering the finite sums
in Eq.~\eqref{eq:MW-raw-derivative} gives
\begin{equation}
 \frac{\mathrm d}{\mathrm dt}\cT^{\rm MW}(t)
 =
 \sum_{K\in\mathcal C_d}
 \sum_{\substack{
 \mathcal J\subseteq\mathcal C_d\setminus\{K\}:\\
 \{K\}\cup\mathcal J\ \text{compatible}}}
 (t-1)^{|\mathcal J|}
 \mathcal W_{\{K\}\cup\mathcal J}.
 \label{eq:MW-component-derivative}
\end{equation}


\subsection{Resummation of the remaining errors}
\label{sec:MW-channel-resummation}
Unlike the case of the optimal recovery threshold, we need to further derive the upper bound for $W_{\mathcal J}$ using another layer of resummation.
Fix a nonempty compatible collection
\(\mathcal J=\{K_1,\ldots,K_j\}\), and let
\(
 J\coloneqq \bigcup_{i=1}^jK_i.
\)
Let \(\partial J\) be the set of vertices outside \(J\) that are
adjacent to \(J\). The requirement that every \(K_i\) be a connected
component of \(\Gamma[S_{B,s}]\) is equivalent to
\begin{equation}
 J\subseteq S_B\cup\supp C_s,
 \qquad
 (S_B\cup\supp C_s)\cap\partial J=\varnothing.
 \label{eq:MW-component-condition}
\end{equation}
The first condition ensures that every vertex of \(J\) belongs to
\(S_{B,s}\), while the second prevents \(J\) from being connected to
any vertex of \(S_{B,s}\setminus J\). Since the \(K_i\) are connected
and pairwise nonadjacent, they are then distinct connected components
of \(\Gamma[S_{B,s}]\).

For a set of physical qubits \(X\), define
\( \Lambda_X\coloneqq \{\alpha\in\Lambda:S_\alpha\cap X\ne\varnothing\}.\)
For \(B\) satisfying Eq.~\eqref{eq:MW-component-condition}, let
\(
 T\coloneqq B\cap\Lambda_J.
\)
The second condition in
Eq.~\eqref{eq:MW-component-condition} implies
\begin{equation}
 T\subseteq\Lambda_J\setminus\Lambda_{\partial J},
 \qquad
 B\cap\Lambda_{\partial J}=\varnothing.
 \label{eq:MW-fixed-locations}
\end{equation}
Hence \(B\) can be written uniquely as
\begin{equation}
     B=T\cup B_{\rm out},
 \qquad
 B_{\rm out}\subseteq
 \Lambda\setminus\Lambda_{J\cup\partial J}.
\end{equation}
Summing over \(B_{\rm out}\), define
\begin{equation}
 F_{J,T}
 \coloneqq 
 \sum_{B_{\rm out}\subseteq
       \Lambda\setminus\Lambda_{J\cup\partial J}}
 F_{T\cup B_{\rm out}}.
 \label{eq:MW-completed-map}
\end{equation}
Equivalently, \(F_{J,T}\) is obtained by composing the maps in their
original circuit order: \(F_\alpha\) is used for \(\alpha\in T\),
\(\id\) for
\(\alpha\in\Lambda_{J\cup\partial J}\setminus T\), and
\(\id+F_\alpha=\Ad_{u_\alpha}\) at every other location.
Thus no commutation between circuit locations is used.

The locations in \(B_{\rm out}\) do not act on \(J\) or
\(\partial J\), so \(J\cap S_B=J\cap S_T\). For the \(K_i\) to remain connected
components of \(\Gamma[S_{B,s}]\), the correction must cover
\(J\setminus S_T\) and avoid \(\partial J\). The syndromes satisfying
these conditions form the set
\begin{equation}
 \Sigma_J(T)\coloneqq 
 \left\{
 s\in\Sigma:
 \begin{array}{l}
 J\setminus S_T\subseteq\supp C_s,\\
 \supp C_s\cap\partial J=\varnothing
 \end{array}
 \right\}.
 \label{eq:MW-syndrome-set}
\end{equation}
Using the notation
\(\cR_\Omega\coloneqq \sum_{s\in\Omega}\cR_s\) for
\(\Omega\subseteq\Sigma\), the sums over \(B_{\rm out}\) and
\(s\in\Sigma_J(T)\) give
\begin{equation}
 \mathcal W_{\mathcal J}
 =
 \sum_{T\subseteq
       \Lambda_J\setminus\Lambda_{\partial J}}
 \cR_{\Sigma_J(T)}\circ F_{J,T}\circ\cE.
 \label{eq:MW-resummation}
\end{equation}
All choices of \(B_{\rm out}\) are summed before any norm is taken.

\subsection{Bounding the resummed expansion}
\label{sec:MW-resummed-bound}

Let \(a\coloneqq 2\sin|\eta|\) for \(|\eta|\le\pi/2\). By
Lemma~\ref{lem:single-fault-norm},
\(\norm{F_\alpha}_\diamond\le a\). Assume \(a\le1\).
In \(F_{J,T}\), each location in \(T\) contributes \(F_\alpha\),
while every other location contributes either \(\id\) or a unitary
channel. Submultiplicativity gives
\begin{equation}
 \norm{F_{J,T}}_\diamond\le a^{|T|}.
 \label{eq:MW-completed-map-bound}
\end{equation}

The Kraus operators of \(\cR_\Omega\) satisfy
\begin{equation}
 \sum_{s\in\Omega}
 (V^\dagger C_s^\dagger\Pi_s)^\dagger
 (V^\dagger C_s^\dagger\Pi_s)
 =
 \sum_{s\in\Omega}\Pi_s
 \le I_P.
 \label{eq:MW-restricted-recovery-effect}
\end{equation}
Thus \(\cR_\Omega\) is completely positive and
trace-nonincreasing, with
\(\norm{\cR_\Omega}_\diamond\le1\), independently of
\(|\Omega|\).

\begin{lemma}
\label{lem:MW-joint-channel-bound}
Let \(\lambda_{\rm MW}\coloneqq 2^D a^{1/(2r)}\). For every nonempty
compatible collection \(\mathcal J\), with
\(J\coloneqq \bigcup_{K\in\mathcal J}K\),
\begin{equation}
 \norm{\mathcal W_{\mathcal J}}_\diamond
 \le\lambda_{\rm MW}^{|J|}.
 \label{eq:MW-joint-channel-bound}
\end{equation}
\end{lemma}

\begin{proof}
By Eqs.~\eqref{eq:MW-resummation} and
\eqref{eq:MW-completed-map},
\(\mathcal W_{\mathcal J}\) is a sum over \(T\),
\(B_{\rm out}\), and \(s\in\Sigma_J(T)\), with
\(B=T\cup B_{\rm out}\). Since the locations in
\(B_{\rm out}\) do not act on \(J\), we have
\(K\cap S_B=K\cap S_T\) for every \(K\in\mathcal J\).
Lemma~\ref{lem:MW-small-components} therefore implies that only
\(T\) satisfying
\begin{equation}
 |K\cap S_T|\ge\frac{|K|}{2}
 \qquad\text{for every }K\in\mathcal J
 \label{eq:MW-local-cover-condition}
\end{equation}
can contribute.

For every \(T\) satisfying
Eq.~\eqref{eq:MW-local-cover-condition}, the sets in
\(\mathcal J\) are disjoint and each location acts on at most
\(r\) qubits. Hence
\begin{equation}
 r|T|
 \ge |J\cap S_T|
 =\sum_{K\in\mathcal J}|K\cap S_T|
 \ge\frac{|J|}{2}.
 \label{eq:MW-T-lower-bound}
\end{equation}
Since the supports within each layer are disjoint, at most
\(|J|\) of them can intersect \(J\). Hence
\(|\Lambda_J|\le D|J|\), and there are at most
\(2^{D|J|}\) choices of \(T\).

For \(0<a\le1\), Eq.~\eqref{eq:MW-T-lower-bound} gives
\(a^{|T|}\le a^{|J|/(2r)}\). Using
Eq.~\eqref{eq:MW-completed-map-bound} and
\(\norm{\cR_{\Sigma_J(T)}}_\diamond\le1\), we obtain
\begin{align}
 \norm{\mathcal W_{\mathcal J}}_\diamond
 &\le
 \sum_{\substack{
 T\subseteq\Lambda_J\setminus\Lambda_{\partial J}:\\
 |K\cap S_T|\ge|K|/2\ 
 \text{ for every }K\in\mathcal J}}
 a^{|T|}
 \nonumber\\
 &\le
 2^{D|J|}a^{|J|/(2r)}
 =\lambda_{\rm MW}^{|J|}.
 \label{eq:MW-joint-channel-estimate}
\end{align}
\end{proof}

Since the sets in \(\mathcal J\) are disjoint,
\(|J|=\sum_{K\in\mathcal J}|K|\). Hence
Lemma~\ref{lem:MW-joint-channel-bound} gives
\(\norm{\mathcal W_{\mathcal J}}_\diamond
\le\prod_{K\in\mathcal J}\lambda_{\rm MW}^{|K|}\).
Define
\begin{equation}
 \Xi_{\rm MW}
 \coloneqq 
 \sum_{K\in\mathcal C_d}\lambda_{\rm MW}^{|K|}.
 \label{eq:MW-Xi}
\end{equation}
After taking norms in
Eq.~\eqref{eq:MW-component-derivative}, we drop the condition that
\(\{K\}\cup\mathcal J\) be compatible and sum over all
\(\mathcal J\subseteq\mathcal C_d\setminus\{K\}\). Using
the subset expansion in
Eq.~\eqref{eq:MW-t-expansion}, we obtain
\begin{align}
 \norm{\frac{\mathrm d}{\mathrm dt}
       \cT^{\rm MW}(t)}_\diamond
 &\le
 \sum_{K\in\mathcal C_d}\lambda_{\rm MW}^{|K|}
 \prod_{K'\in\mathcal C_d\setminus\{K\}}
 \left[1+(1-t)\lambda_{\rm MW}^{|K'|}\right]
 \nonumber\\
 &\le
 \Xi_{\rm MW}e^{(1-t)\Xi_{\rm MW}}.
 \label{eq:MW-derivative-bound}
\end{align}
The second inequality follows from \(1+x\le e^x\) for \(x\ge0\).

Integrating this bound and using
Eq.~\eqref{eq:MW-interpolation-integral}, we obtain
\begin{align}
 \norm{\cR^{\rm MW}\circ\cN_U
       -\cT^{\rm MW}(0)}_\diamond
 &\le
 \int_0^1
 \norm{\frac{\mathrm d}{\mathrm dt}
       \cT^{\rm MW}(t)}_\diamond\,\mathrm dt
 \nonumber\\
 &\le e^{\Xi_{\rm MW}}-1.
 \label{eq:MW-integrated-bound}
\end{align}
By Eq.~\eqref{eq:MW-interpolation-endpoints},
\(\cT^{\rm MW}(0)=c_\Lambda\id_A\). Since
\(\cR^{\rm MW}\circ\cN_U\) is trace preserving, every logical state
\(\sigma\) satisfies
\begin{align}
 |1-c_\Lambda|
 &=
 \left|
 \Tr\left[
 \bigl(
 \cR^{\rm MW}\circ\cN_U-\cT^{\rm MW}(0)
 \bigr)(\sigma)
 \right]
 \right|
 \nonumber\\
 &\le
 \norm{
 \cR^{\rm MW}\circ\cN_U-\cT^{\rm MW}(0)
 }_\diamond.
 \label{eq:MW-small-scalar-bound}
\end{align}
The triangle inequality then gives
\begin{equation}
 \begin{aligned}
 \eps^{\rm MW}(U,\eta)
 &\le
 \frac12\left(
 \norm{
 \cR^{\rm MW}\circ\cN_U-\cT^{\rm MW}(0)
 }_\diamond
 +|c_\Lambda-1|
 \right)\\
 &\le
 \norm{
 \cR^{\rm MW}\circ\cN_U-\cT^{\rm MW}(0)
 }_\diamond.
 \end{aligned}
 \label{eq:MW-recovery-from-tail}
\end{equation}
Combining Eqs.~\eqref{eq:MW-integrated-bound} and
\eqref{eq:MW-recovery-from-tail} with
\(\eps^{\rm MW}(U,\eta)\le1\), we obtain
\begin{equation}
 \eps^{\rm MW}(U,\eta)
 \le
 \min\{1,e^{\Xi_{\rm MW}}-1\}
 \le 2\Xi_{\rm MW}.
 \label{eq:MW-finite-code-bound}
\end{equation}
For \(\Xi_{\rm MW}\le1/2\), the second inequality follows from
\(e^{\Xi_{\rm MW}}-1
\le\Xi_{\rm MW}/(1-\Xi_{\rm MW})
\le2\Xi_{\rm MW}\). For \(\Xi_{\rm MW}>1/2\), it follows from
\(\eps^{\rm MW}(U,\eta)\le1<2\Xi_{\rm MW}\).

\subsection{Coherent error threshold}
\label{sec:MW-threshold}

Equation~\eqref{eq:MW-Xi} is a sum over connected sets of
\(\Gamma\). Since \(\Gamma\) has \(n\) vertices and maximum degree
at most \(\kappa\), Lemma~\ref{lem:connected-set-count} applies with
\begin{equation}
 \chi_{\rm MW}\coloneqq \max\{1,\mathrm e\kappa\},
 \quad
 \rho_{\rm MW}(\eta)
 \coloneqq \chi_{\rm MW}\lambda_{\rm MW}.
 \label{eq:MW-counting-constants}
\end{equation}
Write \(\rho_{\rm MW}\coloneqq \rho_{\rm MW}(\eta)\). For
\(\rho_{\rm MW}<1\), grouping the connected sets by their size gives
\begin{align}
 \Xi_{\rm MW}
 &\le
 n\sum_{j=d}^n
 \chi_{\rm MW}^{j-1}\lambda_{\rm MW}^j
 \nonumber\\
 &\le
 \frac{n}{\chi_{\rm MW}}
 \sum_{j=d}^\infty\rho_{\rm MW}^j
 =
 \frac{n}{\chi_{\rm MW}}
 \frac{\rho_{\rm MW}^d}{1-\rho_{\rm MW}}.
 \label{eq:MW-component-tail}
\end{align}
The assumption \(a\le1\) used above follows from
\(\rho_{\rm MW}<1\), since
\(a<(2^D\chi_{\rm MW})^{-2r}\le1\).

\begin{theorem}[Minimum-weight decoding threshold]
\label{thm:MW-threshold}
Fix positive integers \(w,\ell,D,r\), and suppose the stabilizer
qLDPC family satisfies \(d_n=\Omega(\log n)\). Let
\(\chi_{\rm MW}\) and \(\rho_{\rm MW}(\eta)\) be defined by
Eq.~\eqref{eq:MW-counting-constants}. For every \(n\) and every
\(|\eta|\le\pi/2\) such that
\(\rho_{\rm MW}\coloneqq \rho_{\rm MW}(\eta)<1\),
\begin{equation}
 \sup_U\eps_n^{\rm MW}(U,\eta)
 \le
 \frac{2n}{\chi_{\rm MW}(1-\rho_{\rm MW})}
 \rho_{\rm MW}^{d_n}.
 \label{eq:MW-main-bound}
\end{equation}
A threshold angle may be chosen as
\begin{equation}
 \eta_{\rm th}^{\rm MW}
 =
 \begin{cases}
 \displaystyle
 \arcsin\!\left[
 \frac12
 \left(
 \frac{\mathrm e^{-1/c_d}}
      {2^D\chi_{\rm MW}}
 \right)^{2r}
 \right],
 & d_n\ge c_d\log n,\\[3mm]
 \displaystyle
 \arcsin\!\left[
 \frac{1}{2(2^D\chi_{\rm MW})^{2r}}
 \right],
 & d_n=\omega(\log n).
 \end{cases}
 \label{eq:MW-threshold-angle}
\end{equation}
For every \(|\eta|<\eta_{\rm th}^{\rm MW}\),
\begin{equation}
 \lim_{n\to\infty}
 \sup_U\eps_n^{\rm MW}(U,\eta)=0.
 \label{eq:MW-threshold-limit}
\end{equation}
The supremum is over all allowed coherent error circuits
\(U_n(\eta)\). The decoder depends only on the code and the measured
syndrome.
\end{theorem}

\begin{proof}
Equations~\eqref{eq:MW-finite-code-bound} and
\eqref{eq:MW-component-tail} hold for every allowed \(U\).
Taking the supremum over \(U\) gives
Eq.~\eqref{eq:MW-main-bound}.

At \(\eta=0\), we have \(U=I\), and the correction of minimum-weight decoder
for the zero syndrome is the identity. Recovery is therefore exact.
We henceforth assume
\(0<|\eta|<\eta_{\rm th}^{\rm MW}\).

Suppose first that \(d_n\ge c_d\log n\). The first choice in
Eq.~\eqref{eq:MW-threshold-angle} gives
\(\rho_{\rm MW}<\mathrm e^{-1/c_d}\). Since
\(n\le\mathrm e^{d_n/c_d}\) for all sufficiently large \(n\),
\begin{equation}
 n\rho_{\rm MW}^{d_n}
 \le
 \exp\!\left[
 -d_n\left(
 \log(1/\rho_{\rm MW})-\frac1{c_d}
 \right)
 \right].
 \label{eq:MW-distance-suppression}
\end{equation}
Since
\(\log(1/\rho_{\rm MW})-1/c_d>0\), the right-hand side of
Eq.~\eqref{eq:MW-distance-suppression} decays exponentially in
\(d_n\). Equation~\eqref{eq:MW-main-bound} therefore implies
Eq.~\eqref{eq:MW-threshold-limit}.

If \(d_n=\omega(\log n)\), the second choice in
Eq.~\eqref{eq:MW-threshold-angle} gives \(\rho_{\rm MW}<1\).
Since \(\log n/d_n\to0\) and
\(\log(1/\rho_{\rm MW})>0\),
\begin{equation}
 \log\!\left(n\rho_{\rm MW}^{d_n}\right)
 =
 d_n\left(
 \frac{\log n}{d_n}-\log(1/\rho_{\rm MW})
 \right)
 \longrightarrow-\infty.
 \label{eq:MW-superlog-tail}
\end{equation}
Equation~\eqref{eq:MW-main-bound} again implies
Eq.~\eqref{eq:MW-threshold-limit}.
\end{proof}

\subsection{Local CPTP noise threshold}
\label{sec:MW-local-cptp}

We now consider the local CPTP noise model of
Appendix~\ref{app:local-cptp}, using the same minimum-weight
decoder as above. Let
\(\cN_\Phi\coloneqq \cM\circ\Phi\circ\cE\), where \(\Phi\) satisfies
Eqs.~\eqref{eq:cptp-circuit} and \eqref{eq:cptp-strength}, and define
\begin{equation}
 \eps_n^{\rm MW}(\Phi,p)
 \coloneqq 
 \frac12
 \norm{\cR^{\rm MW}\circ\cN_\Phi-\id_A}_\diamond.
 \label{eq:MW-cptp-error}
\end{equation}

\begin{theorem}[Local CPTP noise threshold]
\label{cor:MW-local-cptp}
Under the assumptions of Theorem~\ref{thm:MW-threshold}, define
\begin{equation}
 \rho_{\rm MW}(p)
 \coloneqq 
 2^D\chi_{\rm MW}(2p)^{1/(2r)}.
 \label{eq:MW-cptp-effective-noise}
\end{equation}
For every \(n\) and every \(0\le p\le1\) such that
\(\rho_{\rm MW}(p)<1\),
\begin{equation}
 \sup_\Phi\eps_n^{\rm MW}(\Phi,p)
 \le
 \frac{2n}
 {\chi_{\rm MW}[1-\rho_{\rm MW}(p)]}
 \rho_{\rm MW}(p)^{d_n}.
 \label{eq:MW-cptp-main-bound}
\end{equation}
A valid threshold value is
\begin{equation}
 p_{\rm th}^{\rm MW}
 =
 \begin{cases}
 \displaystyle
 \frac12
 \left(
 \frac{\mathrm e^{-1/c_d}}
      {2^D\chi_{\rm MW}}
 \right)^{2r},
 & d_n\ge c_d\log n,\\[3mm]
 \displaystyle
 \frac{1}
 {2(2^D\chi_{\rm MW})^{2r}},
 & d_n=\omega(\log n).
 \end{cases}
 \label{eq:MW-cptp-threshold}
\end{equation}
For every \(0\le p<p_{\rm th}^{\rm MW}\),
\begin{equation}
 \lim_{n\to\infty}
 \sup_\Phi\eps_n^{\rm MW}(\Phi,p)=0.
 \label{eq:MW-cptp-threshold-limit}
\end{equation}
The supremum is over all local CPTP noise circuits satisfying
Eqs.~\eqref{eq:cptp-circuit} and \eqref{eq:cptp-strength}. The
decoder uses only the code and the measured syndrome.
\end{theorem}
\begin{proof}
For each location, set
\(F_\alpha\coloneqq \Phi_\alpha-\id_{S_\alpha}\). By
Eq.~\eqref{eq:cptp-strength},
\(\norm{F_\alpha}_\diamond\le2p\).
Expanding \(\Phi_\alpha=\id_{S_\alpha}+F_\alpha\) in circuit order
gives the channel expansion in
Eq.~\eqref{eq:MW-syndrome-expansion}, with
\(\cN_U\) replaced by \(\cN_\Phi\).

Each \(F_B\) is supported on \(S_B\) and can be written as a linear combination of maps \(Y\mapsto E_1YE_2^\dagger\), where \(E_1\) and
\(E_2\) are Pauli operators supported on \(S_B\). These are the only
properties of \(F_B\) used in the proof of
Lemma~\ref{lem:MW-small-components}. Therefore, both statements of the lemma remain valid for local CPTP noise.

The interpolation and cluster resummation techniques apply
unchanged to local CPTP noise, with
\(\cT^{\rm MW}(1)=\cR^{\rm MW}\circ\cN_\Phi\).
In \(F_{J,T}\), every location outside \(T\) contributes either
\(\id\) or \(\id+F_\alpha=\Phi_\alpha\). Since both maps are CPTP,
they have diamond norm one. Hence
Eq.~\eqref{eq:MW-completed-map-bound} holds with \(a=2p\), and
Lemma~\ref{lem:MW-joint-channel-bound} holds with
\(\lambda_{\rm MW}(p)\coloneqq 2^D(2p)^{1/(2r)}\).

The condition \(\rho_{\rm MW}(p)<1\) implies
\(2p<(2^D\chi_{\rm MW})^{-2r}\le1\). Thus the same cluster counting applies with
\(\rho_{\rm MW}=\rho_{\rm MW}(p)\), proving
Eq.~\eqref{eq:MW-cptp-main-bound}.

The threshold and convergence follow by the same argument
as in the proof of Theorem~\ref{thm:MW-threshold}, with
\(\rho_{\rm MW}(\eta)\) replaced by \(\rho_{\rm MW}(p)\).

\end{proof}
\end{appendix}


\end{document}